\documentclass[letterpaper, 10 pt, conference]{ieeeconf}
\IEEEoverridecommandlockouts                              
\usepackage{cite}
\usepackage{amsmath,amssymb,amsfonts}
\usepackage{siunitx}

\usepackage{tikz}
\usetikzlibrary{positioning, arrows.meta, patterns}
\usepackage{tikz-3dplot}

\usepackage{graphicx}
\usepackage{textcomp}
\usepackage{xcolor}
\usepackage{algorithm}
\usepackage{algpseudocode}
\def\BibTeX{{\rm B\kern-.05em{\sc i\kern-.025em b}\kern-.08em
    T\kern-.1667em\lower.7ex\hbox{E}\kern-.125emX}}

\usepackage[top=0.75in, bottom=0.8in, left=0.76in, right=0.76in]{geometry}

\usepackage{tikz}
\usetikzlibrary{shapes.geometric} 

\usepackage{subcaption} 

\newtheorem{assum}{Assumption}
\newtheorem{prop}{Proposition}
\newtheorem{thm}{Theorem}
\newtheorem{lem}{Lemma}
\newtheorem{cor}{Corollary}
\newtheorem{rem}{Remark}

\usepackage{mathrsfs}

\title{\LARGE \bf
Time-Optimal Control of Finite Dimensional Open Quantum Systems via a Model Predictive Strategy
}

\author{Yunyan Lee, Ian R. Petersen~and~Daoyi Dong
\thanks{Yunyan Lee and Ian R. Petersen are with the School of Engineering, The 
Australian National University, Canberra, ACT 2601, Australia. (email: {\tt\small Yun-Yan.Lee@anu.edu.au, ian.petersen@anu.edu.au})
} 
\thanks{Daoyi Dong is with the Australian Artificial Intelligence Institute, Faculty of Engineering and Information Technology, University of Technology Sydney, NSW 2007, Australia.(email: {\tt\small daoyidong@gmail.com})
}
}

\begin{document}
\maketitle

\begin{abstract}
To mitigate dissipative effects from environmental interactions and efficiently stabilize quantum states, time-optimal control has emerged as an effective strategy for open quantum systems. This paper extends the framework by incorporating Positive Operator-Valued Measures (POVMs) into the control process, enabling quantum measurements to guide control updates at each step. To address uncertainties in measurement outcomes, we derive a lower bound on the probability of obtaining a desired outcome from POVM-based measurements and establish stability conditions that ensure a monotonic decrease in the cost function. The proposed method is applied to finite-level open quantum systems, and we also present a detailed analysis of two-level systems under depolarizing, phase-damping, and amplitude-damping channels. Numerical simulations validate the effectiveness of the strategy in preserving coherence and achieving high fidelity across diverse noise environments.
\end{abstract}

\section{INTRODUCTION}
\label{sec:introduction}
Quantum systems play a fundamental role in modern science and technology, enabling advancements in quantum computing, communication, and sensing \cite{nielsen2010quantum, gisin2007quantum, cozzolino2019high, steane1998quantum, pirandola2020advances, djordjevic2022quantum}. These systems use quantum properties such as superposition and entanglement, enabling breakthroughs in high-speed computation and precision sensing \cite{preskill2018quantum}. However, the coupling to the environment induces decoherence and dissipation, reducing quantum state fidelity and limiting performance \cite{wiseman2010quantum, schlosshauer2004decoherence, schlosshauer2019quantum}. To address this issue, quantum control techniques manipulate quantum dynamics to maintain coherence and enhance system performance \cite{dong2022quantum, dong2010quantum, altafini2012modeling, dalessandro2021introduction, James2008H, Mabuchi2008Coherent, james2021optimal, albertini2014time, grivopoulos2008optimal}. Among control techniques, time-optimal control is particularly effective as it minimizes system exposure to environmental noise, preserving coherence and improving reliability \cite{Boscain2021Introduction}.

Time-optimal control aims to transfer a quantum system to a target state as quickly as possible, minimizing decoherence and dissipation \cite{Boscain2021Introduction, Damme2017Robust, Sugny2008ptimal, Sugny2007Time, khaneja2001time, Garon2013Time, Lin2020Time}. A useful tool for designing time-optimal control fields is the Pontryagin Maximum Principle (PMP), which provides necessary conditions for optimality \cite{Boscain2021Introduction, Damme2017Robust, Sugny2007Time}. Complementing the PMP, the quantum speed limit (QSL) sets fundamental bounds on the minimum time required for a quantum state transition, providing theoretical constraints on time-optimal control \cite{deffner2017quantum, caneva2009optimal, gajdacz2015time}. Together, the PMP and QSL establish fundamental constraints for time-optimal control, addressing challenges such as nonlinearity, dissipation, and quantum measurement-induced stochasticity. 

In open quantum systems, where decoherence caused by environmental interactions can severely degrade performance, time-optimal strategies are crucial for preserving quantum state fidelity and ensuring efficient control \cite{Lin2020Time, Sugny2007Time, chen2015near}. Extending these principles to measurement-based control, we introduce a time-optimal strategy that incorporates Positive Operator-Valued Measures (POVMs) \cite{nielsen2010quantum}. This method dynamically adapts the control process based on quantum measurement outcomes and improves fidelity.

Integrating POVMs into time-optimal control enables us to guide real-time control adjustments. This approach is similar to Model Predictive Control (MPC) in classical systems. In classical control, MPC is widely used for optimizing control actions over finite time horizons while ensuring constraint satisfaction \cite{grune2017nonlinear, rawlings2017model, Berberich2021Data, Kohler2021A, Cannon2011Stochastic, Lorenzen2017Constraint}. MPC-based approaches in quantum control have been shown to enhance system stability and mitigate uncertainties \cite{clouatre2022model, goldschmidt2022model, hashimoto2017stability, hashimoto2013probabilistic, humaloja2018linear}. By iteratively updating control inputs based on the current state, MPC achieves robust performance despite uncertainties. Similarly, in quantum systems, POVM-based measurements enable adaptive control updates. This measurement-driven strategy enhances robustness by dynamically optimizing control decisions in response to quantum measurement outcomes, ensuring stability and reliability in open quantum systems.

Unlike classical systems where measurements do not affect the system trajectory, quantum measurements introduce uncertainties and can significantly alter system states \cite{nielsen2010quantum, wiseman2010quantum}. This uncertainty complicates the design and implementation of quantum control strategies. While quantum measurements disrupt system evolution, they can also enhance control effectiveness when strategically integrated \cite{zhang2017quantum, uys2018quantum, Bouten2007An, Mirrahimi2007Stabilizing}. In prior work \cite{lee2024robust}, we applied POVM-based control to a closed two-level quantum system and demonstrated its ability to enhance state stabilization. Here, we propose a time-optimal control strategy that actively incorporates quantum measurements for open quantum systems. By integrating quantum measurements into time-optimal control, our approach mitigates decoherence while maintaining fidelity, making these strategies more viable for practical quantum technologies.

A key challenge is achieving a high probability of successful state transitions after measurement. To address this, we derive a lower bound on the success probability of POVM-based measurements, ensuring reliable state transitions under measurement-based control. We further establish Lyapunov stability conditions, proving that each measurement step reduces the candidate cost function, ensuring stable control dynamics. This guarantees robust stability across open quantum systems and highlights the effectiveness of integrating time-optimal control with POVM-based measurements.

To illustrate the applicability of the proposed framework, we analyze an $N$-level open quantum system. Furhermore, we also examine specific cases, including two-level systems and common decoherence models such as depolarizing and amplitude-damping channels. For systems governed by unitary Lindblad operators, we show how time-optimal control mitigates dissipation and preserves state fidelity. Numerical simulations confirm the strategy’s robustness in mitigating decoherence and achieving high-fidelity state transitions across diverse environmental conditions.

The structure of this paper is as follows. Section~\ref{sec:open_quantum_systems} introduces the fundamental concepts of open quantum systems. Section~\ref{sec:Designing the Controller} presents the design of a time-optimal control strategy, incorporating POVM-based measurements. Section~\ref{sec:probability_bound} derives a lower bound on the success probability of achieving the desired state. Section~\ref{sec:lyapunov_analysis} provides stability analysis based on the Lyapunov criteria. Section~\ref{sec:decoherence_cases} explores specific decoherence models, highlighting stability conditions for various scenarios. Section~\ref{sec:numerical_simulations} presents numerical simulations to validate the proposed approach. Finally, Section~\ref{sec:conclusion} concludes the paper and outlines potential directions for future research.

\section{Problem Formulation}
\label{sec:open_quantum_systems}
A quantum state is described by a density matrix \(\rho\), which is a complex Hermitian matrix satisfying \(\text{Tr}(\rho) = 1\), \(\rho = \rho^\dagger\), and \(\rho \succeq 0\) (i.e., positive semidefinite). For a pure state, the density matrix takes the form \(\rho = |\psi\rangle \langle \psi|\), where \(|\psi\rangle\) is a unit vector in the underlying complex Hilbert space and \(\langle \psi|\) is its Hermitian conjugate. For a Markovian open quantum system, its dynamics of \(\rho\) can be described by the Lindblad equation~\cite{breuer2002theory}:
\begin{align}
\label{eqn:Lindblad_equation}
    \dot{\rho} = -i[H(t), \rho] + \mathbb{D}(\rho),
\end{align}
where \(H(t) = H_0 + H_u\) is the total Hamiltonian. The notation \([A, B] := AB - BA\) denotes the commutator. Here, \(H_0\) is the free Hamiltonian, and \(H_u = \sum_\mu u_\mu(t) H_\mu\) is the control Hamiltonian, with real-valued control inputs \(u_\mu(t)\) and corresponding control operators \(H_\mu\). The dissipative term \(\mathbb{D}(\rho)\) accounts for environmental interactions as follows:
\begin{align}
\label{eqn:D(rho)}
    \mathbb{D}(\rho) = \sum_{i} \gamma_i \left(L_i \rho L_i^\dagger - \frac{1}{2}\{L_i^\dagger L_i, \rho\}\right),
\end{align}
where \(L_i\) are Lindblad operators normalized such that  \(\|L_i\| := \sup_{\|\psi\| = 1} \|L_i \psi\|= 1\), and \(\gamma_i \in \mathbb{R}^+\) are the associated dissipation rates. The notation \(\{A, B\} := AB + BA\) denotes the anti-commutator. If the system is isolated from the environment, i.e., \(\mathbb{D}(\rho) = 0\), the Lindblad equation \eqref{eqn:Lindblad_equation} reduces to the von Neumann equation:
\begin{align}
\label{eqn:nominal_evolution}
    \dot{\rho} = -i[H, \rho],
\end{align}
which describes unitary evolution in closed quantum systems.

Environmental interactions governed by \eqref{eqn:Lindblad_equation} drive the quantum state toward a mixed state, characterized by \(\text{Tr}(\rho^2) < 1\). Once mixed, the state cannot evolve to a pure state through the control Hamiltonian \(H_u\). To preserve quantum coherence and fidelity, it is therefore essential to minimize the duration of exposure to decoherence. Time-optimal control offers a natural framework for this objective.

A time-optimal control problem is formulated as
\begin{align}
\label{eqn: OCP}
    J(\rho(t)) = \min_u \left( \lambda_0 \int_t^{t_f} dt + \mathcal{E}(\rho_{\text{tar}}, \rho(t_f)) \right),
\end{align}
where \(\rho(t)\) denotes the quantum state at time \(t\), \(t_f\) is the final time, and \(u\) represents the control input. The constant \(\lambda_0\) regulates the trade-off between minimizing the evolution time and achieving proximity to the target state \(\rho_{\text{tar}}\).

The terminal error measure is defined as
\begin{align}
\label{eqn: definition cost function}
    \mathcal{E}(\rho_{\text{tar}},\rho(t_f)) = \mathcal{D}^2(\rho_{\text{tar}}, \rho(t_f)),
\end{align}
where \(\mathcal{D}(\rho_{\text{tar}}, \rho(t_f))\) is the trace distance,
\begin{align} 
\label{eq:trace_norm} 
\mathcal{D}(\rho_{\text{tar}}, \rho(t_f)) = \frac{1}{2} \operatorname{Tr}\left( \left| \rho(t_f) - \rho_{\text{tar}} \right| \right),
\end{align} 
which quantifies the distance between quantum states~\cite{nielsen2010quantum}. In the special case where both \(\rho_{\text{tar}} = |\psi_{\text{tar}}\rangle \langle \psi_{\text{tar}}|\) and \(\rho(t_f) = |\psi_f\rangle \langle \psi_f|\) are pure, the terminal error reduces to the infidelity:
\begin{align}
     \mathcal{E}(\rho_{\text{tar}},\rho(t_f)) = 1 - |\langle \psi_f | \psi_{\text{tar}} \rangle|^2.
\end{align}

\section{Control Design}
\label{sec:Designing the Controller}

In this work, we adopt a measurement-based feedback strategy by employing a sequence of POVMs. After each control interval, a POVM is performed, and the resulting post-measurement state is projected back to a pure state for use in the next control step. This procedure forms a feedback loop that helps mitigate decoherence while driving the system toward the target state. The POVMs are constructed from the nominal state solved from the evolution \eqref{eqn:nominal_evolution}.

At each step \(k\), we consider the evolution over a time interval \(T_s(k) = T_{k+1} - T_k\), starting from the current state \(\rho_{T_k}\). The predicted nominal state is denoted by \(\rho_{T_s(k)|T_k}\), obtained by the nominal evolution \eqref{eqn:nominal_evolution}. Based on this state, the POVM is defined as
\begin{align}
   M_k = \left\{ \rho_{T_s(k)|T_k},\; \mathbb{I} - \rho_{T_s(k)|T_k} \right\}.
\end{align}
The measurement yields a post-measurement state \(\rho_{\text{post}}\). To ensure a pure state for the next control step, we project \(\rho_{\text{post}}\) onto the closest pure state with respect to quantum fidelity:
\begin{align}
    |\psi'_{\text{post}}\rangle = \arg\min_{|\psi\rangle} \sqrt{1 - \langle\psi|\rho_{\text{post}}|\psi\rangle}.
\end{align}
The system state is then reset to \(\rho_{T_k} \leftarrow |\psi'_{\text{post}}\rangle \langle \psi'_{\text{post}}|\), and the process is repeated. The complete control procedure is summarized in Algorithm~\ref{alg: qMPC}.

\begin{algorithm}[htbp]
\caption{q-TOPC: Quantum Time-Optimal Predictive Control}
\label{alg: qMPC}
\begin{algorithmic}[1]
\Require Initial state \(\rho_0\), target state \(\rho_{\text{tar}}\), and time sequence \(T = \{T_1, T_2, \ldots, T_N\}\)
\For{\(k = 1\) to \(N-1\)}
    \State Solve \eqref{eqn: OCP} with initial state \(\rho_{T_k}\) to obtain optimal control \(u^\star(\rho_{T_k}, t)\)
    \State Let \(T_s(k) = T_{k+1} - T_k\), and compute state \(\rho_{T_s(k)|T_k}\) via \eqref{eqn:nominal_evolution}
    \State Apply POVM \(M_k = \{\rho_{T_s(k)|T_k},\; \mathbb{I} - \rho_{T_s(k)|T_k}\}\) and obtain \(\rho_{\text{post}}\)
    \State Project to the nearest pure state:
    \[
        |\psi'_{\text{post}}\rangle = \arg \min_{|\psi\rangle} \sqrt{1 - \langle\psi|\rho_{\text{post}}|\psi\rangle}
    \]
    \State Update state: \(\rho_{T_k} \leftarrow |\psi'_{\text{post}}\rangle \langle \psi'_{\text{post}}|\)
\EndFor
\end{algorithmic}
\end{algorithm}

The time-optimal control problem \eqref{eqn: OCP} can be solved using various approaches. In particular, Lin et al.~\cite{Lin2020Time} proposed an analytical solution for two-level quantum systems based on the Pontryagin Maximum Principle (PMP), which leads to bang-bang control strategies. Their result serves as a foundation for our time-optimal formulation and motivates its integration into the MPC framework.

Fig.~\ref{fig: Diagram of qMPC} illustrates the overall structure of the algorithm \ref{alg: qMPC}, q-TOPC. Starting from the initial state \(\rho_0\), the system evolves through a sequence of time intervals \(\{T_1, T_2, \ldots, T_N\}\), with each step involving a time-optimal control action, POVM measurement, and state reset.

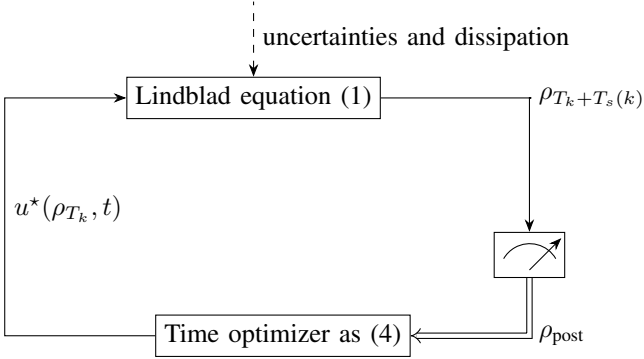
\begin{figure}[htbp]
    \centering
\begin{tikzpicture}[>=Stealth, node distance=20cm] 

\node[draw, rectangle, align=center] (optimizer) at (0,-3) { Time optimizer as \eqref{eqn: OCP}}; 
\node[draw, rectangle, align=center, above left=2.6cm and -3cm of optimizer] (plant) {Lindblad equation \eqref{eqn:Lindblad_equation}};
\node[draw, rectangle, align=center, minimum size=0.6cm, below right=1.5cm and 1.5cm of plant] (POVM) {\qquad};

\draw (POVM.center) ++(0.35,-0.05cm) arc (30:150:0.4cm);
\draw[->] (POVM.center) ++(0,-0.2cm) -- ++(45:0.6cm);

\draw[->] (optimizer.west) -- ++(-2,0) |- node[pos=0.1, right, yshift=30pt] {$u^\star (\rho_{T_k},t)$} (plant.west);
\draw[->] (plant.east) -- +(2,0) -| node[midway, right] {$\rho_{T_k+T_s(k)}$} (POVM.north);
\draw[double equal sign distance, -Implies] (POVM.south) +(0,0) |- node[midway, right] {$\rho_\text{post}$}  (optimizer.east);
\draw[dashed, <-] (plant.north) -- ++(0,1) node[midway, right, align=center] {uncertainties and dissipation};

\end{tikzpicture}
\caption{
Diagram of q-TOPC: The algorithm begins with an initial state \(\rho_0\) and a target state \(\rho_{\text{tar}}\). In each iteration, the performance index \(J (\rho(t))\) is optimized to compute the optimal control input \(u^\star(\rho_{T_k}, t)\). The system evolves according to the nominal dynamics, resulting in a state \(\rho_{T_s(k)|T_k}\). A POVM is then performed, yielding the post-measurement state \(\rho_{\text{post}}\). The closest pure state \(|\psi'_{\text{post}}\rangle\) is identified, and the system state is updated. This process is repeated, with each post-measurement state \(|\psi'_{\text{post}}\rangle\) being fed back into the optimization until it converges toward the desired target state \(\rho_{\text{tar}}\).
}
    \label{fig: Diagram of qMPC}
\end{figure}

In the following, we show that the post-measurement state \(\rho_{\text{post}}\) coincides with the nominal state \(\rho_{T_s(k)|T_k}\) with high probability. Based on this, we establish that each iteration in Algorithm~\ref{alg: qMPC} leads to a monotonic decrease in the cost function \(J\), thereby ensuring Lyapunov stability.

\section{The Minimum Success Probability}
\label{sec:probability_bound}
To analyze the stability of the system, we first quantify the probability that the measurement outcome corresponds to the nominal state under POVM measurements. This probability is given by
\begin{align}
\label{eqn:prob_POVM}
    p = 1-\text{Tr} \left( \rho_{t + T_s} \rho_{T_s|t} \right),
\end{align}
where \(p \in [0, 1]\) represents the deviation from the nominal state at each time step. For notational simplicity, we use \(\rho_t\) to denote the system state at time \(t \in [0, t_f]\), and \(\rho_{T_s|t}\) to denote the nominal state obtained by evolving \(\rho_t\) over the interval \(T_s\) according to the nominal dynamics in \eqref{eqn:nominal_evolution}. In the presence of uncertainties and dissipation, the actual system state at time \(t + T_s\) is denoted by \(\rho_{t + T_s}\), which evolves under the Lindblad equation \eqref{eqn:Lindblad_equation}.

To account for model uncertainties and environmental interactions, we consider systems affected by both Hamiltonian perturbations and dissipation. Specifically, the system is governed by an uncertain Hamiltonian \(H_\Delta\), with its operator norm bounded by \(\|H_\Delta\| \leq \bar{\Delta}\), and a dissipative term \(\mathbb{D}(\rho)\) as described in \eqref{eqn:Lindblad_equation}. Following the stochastic unraveling framework presented in \cite{yip2018quantum}, we distill the deterministic evolution from the Lindblad equation, as detailed in Appendix~\ref{sec: Appendix B}. This formulation leads to Theorem~\ref{thm: general min suc prob}, which provides a lower bound on the success probability defined in \eqref{eqn:prob_POVM}.

\begin{thm}
\label{thm: general min suc prob}
Consider an open quantum system subject to an uncertain Hamiltonian \( H_\Delta \) and a dissipation term \( \mathbb{D}(\rho) \), with dynamics governed by
\begin{align}
\label{eqn:Lindblad_uncertain}
    \dot{\rho}(t) = -i[H(t) + H_\Delta, \rho(t)] + \mathbb{D}(\rho(t)),
\end{align}
where the uncertainty satisfies \( \|H_\Delta\| = \Delta \|H'_\Delta\| \), with \( \|H'_\Delta\| = 1 \) and \( \Delta \leq \bar{\Delta} \). Assume that the overall dissipation rate encoded in \( \mathbb{D}(\rho) \) is uniformly bounded by \( \bar{\gamma} \). Then, the probability \( p \) of deviating from the nominal state at time \( t + T_s \) is bounded above by
\begin{align}  
p \leq \epsilon T_s,
\end{align}  
where \( \epsilon = 2\bar{\Delta} + \bar{\gamma} \).
\end{thm}

\begin{proof}
Based on the stochastic unraveling method \cite{yip2018quantum}, the Lindblad equation \eqref{eqn:Lindblad_uncertain} can be decomposed into a deterministic component and a jump component. Specifically, the deterministic part of the evolution can be written as
\begin{align}
    \tilde{\rho}_d = V_{\mathrm{eff}} \rho_t V_{\mathrm{eff}}^\dagger,
\end{align}
where \( V_{\mathrm{eff}} \) is the non-unitary evolution operator generated by the effective non-Hermitian Hamiltonian \( \tilde{H}_{\mathrm{eff}} \), defined as
\begin{align}
    V_{\mathrm{eff}} = \exp\left( -i \int_t^{t+T_s} \tilde{H}_{\mathrm{eff}}(s) \, ds \right),
\end{align}
with
\begin{align}
    \tilde{H}_{\mathrm{eff}}(s) = H(s) + H_\Delta - \frac{i}{2} \sum_i \gamma_i(s) L_i^\dagger L_i.
\end{align}

Following \eqref{eqn: deterministic and jump} and \eqref{eqn: no jump traj prob} in Appendix~\ref{sec: Appendix B}, the average system state over all quantum trajectories at time \( t + T_s \) is
\begin{align}
\label{eqn: state gotten from the quantum trajectory}
    \rho_{t+T_s} = \exp\left(- \int_{t}^{t+T_s} \sum_i \gamma_i (t)  d s \right) \rho_d +  \sum_i p_i \rho_{\text{jump},i},
\end{align}
where the normalized deterministic state is
\begin{align}
    \rho_d = \exp \left (\int_t^{t+T_s} \sum_i \gamma_i(s) ds \right) \tilde{\rho}_d,
\end{align}
and the jump probabilities satisfy
\begin{align}
    \sum_i p_i = 1 - \exp \left( - \int_t^{t+T_s} \sum_i \gamma_i(s) ds \right),
\end{align}
with \( \rho_{\text{jump},i} = |\psi_{\mathrm{jump},i}\rangle \langle \psi_{\text{jump},i}| \) denoting the jump states.

In the worst-case scenario, we assume that all post-jump trajectories yield states orthogonal to the nominal state \( \rho_{T_s|t} \), i.e., \( \text{Tr}(\rho_{T_s|t} \rho_{\text{jump},i}) = 0 \) for all \( i \). Consequently, based on \eqref{eqn: state gotten from the quantum trajectory}, the overall success probability is bounded below by the contribution from the deterministic evolution:
\begin{align}
    \text{Tr}(\rho_{t+T_s} \rho_{T_s|t}) \geq \text{Tr}(\tilde{\rho}_d \rho_{T_s|t}).
\end{align}

To evaluate this term, we denote \( \tilde{\rho}_d = |\tilde{\psi}_d\rangle \langle \tilde{\psi}_d| \) and \( \rho_{T_s|t} = |\psi_{T_s|t}\rangle \langle \psi_{T_s|t}| \), and obtain
\begin{align}
    \text{Tr}(\tilde{\rho}_d \rho_{T_s|t}) = |\langle \psi_{T_s|t} | \tilde{\psi}_d \rangle|^2.
\end{align}

To estimate the time variation of this quantity, we consider \( |\langle \phi(t^+) | \tilde{\psi}(t^+) \rangle|^2 \), where \( |\phi(t^+)\rangle \) and \( |\tilde{\psi}(t^+)\rangle \) evolve under the nominal and effective non-Hermitian dynamics, respectively. Taking the time derivative yields:
\begin{align}
\label{eqn: derivative of expectation}
\begin{aligned}
    &\frac{d}{dt^+} |\langle \phi(t^+) | \tilde{\psi}(t^+) \rangle|^2 \\ 
    &= 2\,\mathrm{Re} \left[ \left\langle \phi(t^+) \middle| \frac{d}{dt^+} \tilde{\psi}(t^+) \right\rangle \langle \tilde{\psi}(t^+) | \phi(t^+) \rangle \right].
\end{aligned}
\end{align}

The derivative term can be expressed as:
\begin{align}
\begin{aligned}
    &\left\langle \phi(t^+) \middle| \frac{d}{dt^+} \tilde{\psi}(t^+) \right\rangle \\ 
    &= - \langle \psi(t) | U(t^+, t)^\dagger (i H_\text{eff}(t^+)) V_\text{eff}(t^+, t) | \psi(t) \rangle \\ 
    &\quad - \frac{1}{2} \sum_i \gamma_i(t^+) \langle \psi(t) | U(t^+, t)^\dagger L_i^\dagger L_i V_\text{eff}(t^+, t) | \psi(t) \rangle.
\end{aligned}
\end{align}

Using the relation \( \langle \tilde{\psi}(t^+) | \phi(t^+) \rangle = \langle \psi(t^+) | \phi(t^+) \rangle \), where \( |\psi(t^+)\rangle \) is the normalized version of \( |\tilde{\psi}(t^+)\rangle \), we obtain:
\begin{align}
\begin{aligned}
    &\frac{d}{dt^+} |\langle \phi(t^+) | \tilde{\psi}(t^+) \rangle|^2 \\ 
    &\geq 2\,\mathrm{Re} \left[ \left\langle \phi(t^+) \middle| \frac{d}{dt^+} \tilde{\psi}(t^+) \right\rangle \langle \psi(t^+) | \phi(t^+) \rangle \right].
\end{aligned}
\end{align}

To bound this expression, we use the norm bounds \( \|H_\Delta\| = \Delta \|H'_\Delta\| \) with \( \|H'_\Delta\| = 1 \). Noting that \( \|\psi\| = 1 \) and \( \|\phi\| = 1 \), and that \( \|L_i\| = 1 \) by definition, we obtain:
\begin{align}
    \frac{d}{dt^+} |\langle \phi(t^+) | \tilde{\psi}(t^+) \rangle|^2 \geq -2\bar{\Delta} - \bar{\gamma}.
\end{align}

Integrating over the time interval \( [t, t+T_s] \), we arrive at the lower bound:
\begin{align}
    \text{Tr}(\rho_{t+T_s} \rho_{T_s|t}) \geq 1 - (\bar{\gamma} + 2\bar{\Delta}) T_s.
\end{align}
\end{proof}

\begin{rem}
In addition to the lower bound for open finite-level systems given in Theorem~\ref{thm: general min suc prob}, a corresponding result for closed systems with time-independent Hamiltonians is provided in Appendix~\ref{sec: Appendix A}, based on the approach of~\cite{childs2019nearly}.  
The case of proportional Hamiltonian noise, i.e., \(H_\Delta \propto H_0\), is also addressed for finite-level systems in~\cite{berberich2024robustness}.
\end{rem}

Theorem~\ref{thm: general min suc prob} provides a lower bound on the probability of obtaining the nominal state \(\rho_{T_s|t}\). This probability estimate will serve as a key component in establishing the Lyapunov stability of Algorithm~\ref{alg: qMPC}, which will be discussed in the next section.

\section{Stability Analysis of the Time-Optimal Control with POVM}
\label{sec:lyapunov_analysis}
To prove stability under the given cost function, we assume that the system is controllable from the initial state \( \rho_0 \) to the target state \( \rho_{\text{tar}} \) under the dissipation-free dynamics described by \eqref{eqn:nominal_evolution}. 

\begin{assum}
\label{ass: controllability}
    There exists a control law under which the system governed by \eqref{eqn:nominal_evolution} can be steered from the initial quantum state \(\rho_0\) to the target state \(\rho_{\text{tar}}\).
\end{assum}

\begin{prop}
\label{prop: Lyapunov stable}
Suppose the controller satisfies Assumption~\ref{ass: controllability}, and consider the optimal control problem \eqref{eqn: OCP}. If the measured quantum state coincides with the nominal state, then each iteration of Algorithm~\ref{alg: qMPC} satisfies
\begin{align}
    \Delta J := J(\rho_{k+1}) - J(\rho_k) \leq 2\sqrt{\bar{\epsilon} T_s} - \lambda_0 T_s.
\end{align}
In particular, if
\begin{align}
    \bar{\epsilon} \leq \frac{\lambda_0^2 T_s}{4},
\end{align}
then \( \Delta J \leq 0 \), and the cost function \( J \) is non-increasing along the closed-loop trajectory. Hence, \( J \) serves as a discrete-time Lyapunov function, establishing Lyapunov stability of the closed-loop system under Algorithm~\ref{alg: qMPC}.
\end{prop}

\begin{proof}
For a given initial state \(\rho_t\), let \(t^*_f(\rho_t)\) denote the final time obtained by solving the optimal control problem \eqref{eqn: OCP}. Similarly, when the same optimal control sequence is applied to the nominal state \(\rho_{T_s|t}\), the corresponding final time is denoted by \(t^*_f(\rho_{T_s|t})\).

Let \(\rho_f\) represent the state that starts from \(\rho_{t}\) and evolves under Lindblad equation \eqref{eqn:Lindblad_equation} at optimal final time $t^*_f(\rho_{t})$ and \(\rho_f'\) represent the state that starts from \(\rho_{T_s|t}\) and evolves under the same dynamics at optimal final time $t^*_f(\rho_{T_s|t})$.

The difference in the functional values can be expressed as:
\begin{align}
\begin{aligned}
      &J(\rho_{T_s|t}) - J(\rho_t) \\
      =& \lambda_0\left[\int_{t+T_s}^{t^*_f(\rho_{T_s|t})} dt +  \mathcal{E}(\rho_{\text{tar}},\rho_f') \right] \\
      &- \left[\lambda_0\int_t^{t+T_s} dt + \int_{t+T_s}^{t_f^\star (\rho_t)} dt + \mathcal{E}(\rho_{\text{tar}},\rho_f) \right].
\end{aligned}
\end{align}

Next, consider a suboptimal control sequence for \( \rho_{T_s|t} \) with \( u(\rho_t, t) \) derived from the optimal control sequence for \( \rho_t \) over the time interval from \( t + T_s \) to \( t_f^\star (\rho_t) \), which corresponds to the final state $\tilde \rho_f$. Then, we have:
\begin{align}
\label{eqn: suboptimal difference}
\begin{aligned}
      &J(\rho_{T_s|t}) - J(\rho(t)) \\
      &\leq -\lambda_0T_s + \mathcal{E}(\rho_{\text{tar}},\tilde \rho_f)  - \mathcal{E}(\rho_{\text{tar}},\rho_f)\\
      &\leq (\mathcal{D}(\rho_{\text{tar}},\tilde \rho_f) +\mathcal{D}(\rho_{\text{tar}},\rho_f))\mathcal{D}(\tilde \rho_f,\rho_f)  -\lambda_0T_s \\
      &\leq 2\mathcal{D}(\tilde \rho_f,\rho_f)  -\lambda_0T_s \\
      &\leq  2\mathcal{D}(\mathcal{E}_T(\rho_{T_s| t}),(\mathcal{E}_T(\rho_{t+T_s})) -\lambda_0T_s \\
      &\leq  2\mathcal{D}( \rho_{T_s| t} , \rho_{t+T_s} )  -\lambda_0T_s \\
      &\leq 2\sqrt{\bar \epsilon T_s}  -\lambda_0T_s,
\end{aligned}
\end{align}
where \( \mathcal{D}(\cdot,\cdot) \) denotes the trace norm for quantum states as defined in \eqref{eq:trace_norm}, and the channel \( \mathcal{E}_T \) follows \eqref{eqn:Lindblad_equation} with \( T = t_f^\star (\rho_{t}) - t - T_s \). The control signal used in \( \mathcal{E}_T \) is derived from the optimal cost function with the initial state \( \rho_{t+T_s} \). By the fact that trace-preserving quantum operations are contracting \cite[Theorem 9.2]{nielsen2010quantum}, it follows that for any completely positive trace-preserving (CPTP) channel,
\begin{align}
\mathcal{D}(\mathcal{E}_T(\rho_{T_s| t}),\mathcal{E}_T(\rho_{t+T_s})) \leq \mathcal{D}( \rho_{T_s| t} , \rho_{t+T_s} ).
\end{align}

The last inequality of \eqref{eqn: suboptimal difference} can be obtained from Theorem \ref{thm: general min suc prob} since, in the probability of $\text{Tr} (\rho_{t+T_s}\rho_{T_s|t})$, $\rho_{T_s|t}$ is a pure state. Thus, the fidelity is
\begin{align}
    F(\rho_{t+T_s}, \rho_{T_s|t}) = \sqrt{\text{Tr} (\rho_{t+T_s} \rho_{T_s|t})}.
\end{align}
Also, it is known from \cite[Chapter 9]{nielsen2010quantum} that
\begin{align}
    \mathcal{D}( \rho_{T_s|t}, \rho_{t+T_s} ) \leq \sqrt{1 - F( \rho_{T_s|t}, \rho_{t+T_s} )^2}
\end{align}
and consider the Lindblad equation in the optimal cost function \eqref{eqn: OCP} with an upper bound $\bar \epsilon$. Thus, to ensure Lyapunov stability, the dissipative term should satisfy:
\begin{align}
    \bar \epsilon \leq \frac{\lambda_0^2 T_s}{4}.
\end{align}
\end{proof}

To maintain high fidelity, we impose the terminal-cost bound 
\begin{align}
    \mathcal{E} \leq \bar{\epsilon} N T_s,
\end{align}
where time horizon \(N T_s\) is at least as long as the minimal time required—under Assumption~\ref{ass: controllability}—for the nominal dynamics to reach the target state.  
Theorem~\ref{thm: general min suc prob} then guarantees a success probability of at least \(1 - \bar{\epsilon} N T_s\), leading to the following corollary.

\begin{cor}
\label{cor: Lyapunov stable}
Let the controller satisfy Assumption~\ref{ass: controllability}, and consider the OCP defined in~\eqref{eqn: OCP}.  
If the terminal cost satisfies 
\begin{align}
    \mathcal{E} \leq \bar{\epsilon} N T_s,
\end{align}
and
\begin{align}
    \bar{\epsilon} \leq \frac{\lambda_0}{2\sqrt{N}},
\end{align}
then \( \Delta J \leq 0 \), and the cost function in~\eqref{eqn: OCP} is non-increasing under Algorithm~\ref{alg: qMPC}. Hence, it serves as a discrete-time Lyapunov function and establishes Lyapunov stability of the closed-loop system.
\end{cor}

\begin{proof}
Following the proof of Proposition \ref{prop: Lyapunov stable}, in the suboptimal control case, the difference in the cost function can be expressed as:
\begin{align}
\label{eqn: constraint Lyapunov stable}
\begin{aligned}
      &J(\rho_{T_s|t}) - J(\rho(t)) \\
      &\leq -\lambda_0T_s + \mathcal{E}(\rho_{\text{tar}},\tilde \rho_f)  - \mathcal{E}(\rho_{\text{tar}},\rho_f)\\
      &\leq \left(\mathcal{D}(\rho_{\text{tar}},\tilde \rho_f) +\mathcal{D}(\rho_{\text{tar}},\rho_f)\right)\mathcal{D}(\tilde \rho_f,\rho_f)  -\lambda_0T_s\\
      &\leq 2\sqrt{N}\bar\epsilon T_s -\lambda_0T_s.
\end{aligned}
\end{align}

Since any terminal cost function is constrained by \(\epsilon N T_s\), we obtain the final inequality in \eqref{eqn: constraint Lyapunov stable}. Therefore, to satisfy Lyapunov stability, we require:
\begin{align}
    \bar\epsilon \leq \frac{\lambda_0}{2\sqrt{N}}.
\end{align}
\end{proof}

When the dissipative coefficient \(\bar{\gamma}\) is sufficiently small, it is reasonable to design the controller based solely on the nominal dynamics~\eqref{eqn:nominal_evolution}.

\begin{cor}
\label{cor: Lyapunov stable nominal}
Consider a system governed entirely by the nominal dynamics~\eqref{eqn:nominal_evolution}, including both the OCP~\eqref{eqn: OCP} and the POVM-based state evolution.  
Under Assumption~\ref{ass: controllability}, each iteration of Algorithm~\ref{alg: qMPC} yields \( \Delta J \leq -\lambda_0 T_s \). Hence, the cost function in~\eqref{eqn: OCP} serves as a discrete-time Lyapunov function and guarantees Lyapunov stability of the closed-loop system.
\end{cor}

\begin{proof}
    Following the approach in Proposition \ref{prop: Lyapunov stable}, consider a suboptimal cost function \( J(\rho_{T_s|t}) \) for the state after the time interval \(T_s\):
    \begin{align}
    \begin{aligned}
           &J(\rho_{T_s|t}) - J(\rho_t) \\
           &\leq -\lambda_0T_s  + \mathcal{C}(\rho_{\text{tar}},\tilde \rho_f)  - \mathcal{C}(\rho_{\text{tar}},\rho_f)\\
           &\leq -\lambda_0T_s,
    \end{aligned}
    \end{align}
where \( \tilde \rho_f = \rho_f \) since we consider a suboptimal solution and the system is without dissipation. Thus, starting from the same initial state, both the suboptimal and optimal solutions lead to the same final state \( \rho_f \).

Therefore, the cost function \( J(\rho(t)) \) decreases over time, ensuring that the system is Lyapunov stable.
\end{proof}

To conclude this section, we extend the stability result from Proposition~\ref{prop: Lyapunov stable}, which assumes perfect measurement.  
The following theorem establishes that the closed loop system corresponding to Algorithm~\ref{alg: qMPC} remains Lyapunov stable in expectation under a specific condition, as stated in Theorem~\ref{thm: POVM Lyapunov stable}.

\begin{thm}
\label{thm: POVM Lyapunov stable}
Under Algorithm~\ref{alg: qMPC}, if \( \bar{\epsilon} \) satisfies
\begin{align}
    \bar{\epsilon} T_s + 2\sqrt{\bar{\epsilon} T_s} - \lambda_0 T_s \leq 0,
\end{align}
then each iteration of Algorithm~\ref{alg: qMPC} is Lyapunov stable in expectation.
\end{thm}

\begin{proof}
Consider a quantum system evolving under a dissipative quantum operation. Let $\rho_{t+T_s}$ represent the state at time $t + T_s$ given the state at time $t$, and let $\{\rho^\perp_{i, T_s|t}\}$ be a set of states orthogonal to $\rho_{T_s|t}$. Assume that the evolution can be decomposed into a nominal evolution with probability $1-p$ and transitions to orthogonal states with probabilities $\{p_i\}$, where $\sum_i p_i = p \leq \bar{\epsilon}T_s$.  Then, the expected change in the cost function over a single time step can be bounded as follows:
\begin{align}
\label{eq:cost_change}
\begin{aligned}
&\mathbb{E}[J(\rho_{t+T_s}) - J(\rho(t))] \\
&= (1-p) J(\rho_{T_s|t}) + \sum_i p_i J(\rho^\perp_{i, T_s|t}) - J(\rho(t)) \\
&= (1-p)\lambda_0(T-T_s) + \sum_i p_i \lambda_0 (T-T_s) - \lambda_0 T  \\
&\quad + (1-p)\mathcal{C}(\rho_{\text{tar}},\tilde \rho_f)+ \sum_i p_i\mathcal{C}(\rho_{\text{tar}},\tilde \rho^\perp_{i,f}) -\mathcal{C}(\rho_{\text{tar}}, \rho_f)\\
&= -\lambda_0 T_s +\mathcal{C}(\rho_{\text{tar}},\tilde \rho_f)-\mathcal{C}(\rho_{\text{tar}}, \rho_f) \\
&\quad-p \mathcal{C}(\rho_{\text{tar}},\tilde \rho_f)+ \sum_i p_i\mathcal{C}(\rho_{\text{tar}},\tilde \rho^\perp_{i,f})\\
&\leq -\lambda_0 T_s + 2\sqrt{\bar{\epsilon}T_s} + \bar{\epsilon}T_s.
\end{aligned}
\end{align}
The last inequality follows from inequality \eqref{eqn: suboptimal difference} and the fact that the cost function is bounded by $1$. Therefore, when \( -\lambda_0T_s + 2\sqrt{\bar\epsilon T_s} + \bar\epsilon T_s \leq 0 \), we can ensure that Algorithm \ref{alg: qMPC} is Lyapunov stable in expectation.
\end{proof}

When the state follows the nominal trajectory for \(N\) steps, Algorithm~\ref{alg: qMPC} terminates once the optimal time falls below the step size \(T_s\).  
To estimate the probability of reaching the target within \(N\) steps, we refer to the analysis in \cite{lee2024robust}.  
Given the minimum success probability from Theorem~\ref{thm: general min suc prob}, with lower bound \(1 - \bar{\epsilon} T_s\), we define \(P_{\textnormal{tar}}(N)\) as the probability that the system reaches the target state by step \(N\).  
A lower bound is given by
\begin{align}
    P_{\textnormal{tar}}(N) \geq 1 - \bar{\epsilon} T_s \sum_{l=1}^{L} [1 - \bar{\epsilon} T_s]^{l-1} F_{N-l},
\end{align}
where \(F_k\) denotes the probability that the system fails to reach the target by step \(k\).  
This estimate leads to the convergence analysis below.

Following \cite{lee2024robust}, the convergence rate of \(F_N\), as defined in \cite{ortega2000iterative}, is
\begin{align}
    \eta = \sup \lim_{N \to \infty} (F_N)^{\frac{1}{N}},
\end{align}
and we establish the bound for convergence rate \(\eta\).

\begin{prop}
\label{prop:convergence_rate}
The convergence rate \(\eta\) is given in the following three cases:
\begin{enumerate}
    \item When \(1 - \bar{\epsilon} T_s < \frac{L}{L+1}\),
    \begin{align}
        \eta = \min\left\{ 1 - \alpha, \frac{2L}{L+1} - (1 - \bar{\epsilon} T_s) \right\},
    \end{align}
    where \(\alpha = \bar{\epsilon} T_s [1 - \bar{\epsilon} T_s]^L\).
    
    \item When \(1 - \bar{\epsilon} T_s > \frac{L}{L+1}\),
    \begin{align}
        \eta = \frac{2L}{L+1} - (1 - \bar{\epsilon} T_s).
    \end{align}
    
    \item When \(1 - \bar{\epsilon} T_s = \frac{L}{L+1}\),
    \begin{align}
        \eta = 1 - \bar{\epsilon} T_s.
    \end{align}
\end{enumerate}
\end{prop}

The proof follows the derivation in~\cite{lee2024robust}.

In all cases, the convergence rate satisfies \(\eta < 1\), implying that
\begin{align}
    \lim_{N \to \infty} P_{\textnormal{tar}}(N) = 1.
\end{align}

\section{Impact of Different Types of Decoherence}
\label{sec:decoherence_cases}

In quantum mechanics, different types of decoherence may be particularly relevant depending on the application.  
This section presents the minimum success probability of obtaining the nominal state, as stated in Theorem~\ref{thm: general min suc prob}, under various decoherence models, including a tighter bound for two-level quantum systems.

\subsection{Analysis of Depolarizing Decoherence}
\label{sec: Depolarizing Decoherence}

In an $N$-dimensional quantum system, a depolarizing decoherence channel is defined as \cite{nielsen2010quantum}:
\begin{align}
\label{eqn: depolarizing_decoherence_channel}
    \mathcal{E}(\rho) = p_D \frac{\mathbb{I}}{d} + (1-p_D)\rho,
\end{align}
where $p_D$ is the depolarizing probability, and $d$ is the dimension of the system. Representing the density matrix in terms of the generalized Bloch vector \cite{yang2013exploring}, we have:
\begin{align}
    \rho = \frac{1}{d} + \frac{1}{2} \sum_i x_i X_i,
\end{align}
where $x_i$ are the Bloch vector components, and $X_i$ are the corresponding generators.

Now, we calculate the probability:

\begin{align}
\begin{aligned}
       &\text{Tr}(\rho_{T_s|t} \rho_{t+T_s}) = \text{Tr}\left(\left(\frac{1}{d} + \frac{1}{2} \sum_i x_i(T_s|t) X_i\right) \right. \\
       &\qquad \left. \left(\frac{1}{d} + e^{-\gamma T_s} \frac{1}{2} \sum_i x_i(T_s|t) X_i\right)\right) \\
    &= \text{Tr}\left(\frac{\mathbb{I}}{d^2} + e^{-\gamma T_s} \frac{1}{4} \sum_{i,j} x_i(T_s|t) x_j(T_s|t) X_i X_j \right) \\
    &= \frac{1}{d} + \left(1 - \frac{1}{d}\right) e^{-\gamma T_s}.
\end{aligned}
\end{align}
Here, we assume that $p_D = 1 - e^{-\gamma T_s}$, and the last term is derived from the definition of the matrix $X$ for transferring the density matrix in the Bloch vector representation \cite{yang2013exploring, Kimura2003The}.

In depolarizing decoherence case, the system Hamiltonian and depolarizing decoherence channel can be swapped. Following the subpotimal Lyapunov decrease as \eqref{eqn: suboptimal difference}, we have

\begin{align}
\begin{aligned}
      &J(\rho_{T_s|t}) - J(\rho(t)) \\
      &\leq -\lambda_0T_s + \mathcal{C}(\rho_{\text{tar}},\tilde \rho_f)  - \mathcal{C}(\rho_{\text{tar}},\rho_f)\\
      &\leq (\mathcal{D}(\rho_{\text{tar}},\tilde \rho_f) +\mathcal{D}(\rho_{\text{tar}},\rho_f))\mathcal{D}(\tilde \rho_f,\rho_f)  -\lambda_0T_s \\
      &\leq 2\mathcal{D}(\tilde \rho_f,\rho_f)  -\lambda_0T_s \\
      &\leq  2\mathcal{D}(\mathcal{E}_T(\rho_{T_s| t}),(\mathcal{E}_T(\rho_{t+T_s})) -\lambda_0T_s \\
      &\leq  2\mathcal{D}( \rho_{T_s| t} , \rho_{t+T_s} )  -\lambda_0T_s \\
      &= \sqrt{2p_D(1-\frac{1}{d})}  -\lambda_0T_s.
\end{aligned}
\end{align}
Since \( \rho_{T_s| t} \) must commute with \( \rho_{t+T_s} \) due to the depolarizing evolution, we can directly calculate 
\begin{align}
    \mathcal{D}( \rho , \mathcal{E}(\rho) ) = \dfrac{1}{2} \sqrt{ 2 p_D \left( 1 - \dfrac{1}{d} \right) }
\end{align}
for any pure state \(\rho\) using the depolarizing decoherence channel \eqref{eqn: depolarizing_decoherence_channel} \cite[Chapter 9]{nielsen2010quantum}. Therefore, when
\begin{align}
\label{eqn:decoherence_Lyapunov_decreaing}
    \sqrt{ 2 p_D \left( 1 - \dfrac{1}{d} \right) } - \lambda_0 T_s \leq 0,
\end{align}
the cost function serves as a Lyapunov function and guarantees stability for the system under depolarizing decoherence.

Similarly, to achieve stability in expectation, the following condition should be satisfied:
\begin{align}
\label{eqn:decoherence_Lyapunov_decreaing_expection}
    -\lambda_0 T_s + \sqrt{ 2 p_D \left( 1 - \dfrac{1}{d} \right) } + \left( 1 - \dfrac{1}{d} \right)p_D \leq 0.
\end{align}

\subsection{Uniform Dissipation: \(L^\dagger L = \gamma(t) \mathbb{I}\)}
\label{sec:Uniform Dissipation}

In the uniform dissipative case, the Lindblad operator satisfies  
\begin{align}
    L^\dagger L = \gamma(t) \mathbb{I}.
\end{align}
This condition appears in certain decoherence models, such as phase-damping decoherence. Under this assumption, a refined lower bound on the minimum probability of obtaining the nominal state can be derived by applying the stochastic unraveling of the quantum system \cite{yip2018quantum}.
\begin{prop}
    Given that the Lindblad operator satisfies \(L^\dagger L = \gamma(t) \mathbb{I}\) and \(0\leq \gamma (t) \leq \bar \gamma\), the trace of the product of the evolved state \(\rho_{t+T_s}\) and the predicted state \(\rho_{T_s|t}\) satisfies the relation:
    \begin{align}
        \text{Tr} (\rho_{t+T_s} \rho_{T_s|t}) \geq e^{-\bar\gamma T_s}.
    \end{align}
\end{prop}

\begin{proof}
    To analyze this case, we examine the Lindblad equation using stochastic unraveling \cite{yip2018quantum}. Since the Lindblad operator involves the identity matrix, it commutes with the Hamiltonian. Consequently, following \eqref{eqn: deterministic evolution} in Appendix \ref{sec: Appendix B}, the deterministic evolution can be expressed as:
    \begin{align}
        \tilde{\rho}_d = V_{\text{eff}} \rho_t V^\dagger_{\text{eff}},
    \end{align}
    where 
    \[
    V_{\text{eff}} = \exp \left( -i \int_{t}^{t+T_s} H(s) - \frac{1}{2} \gamma(s) \mathbb{I} \, ds \right).
    \]
    
    Following the proof of Theorem \ref{thm: general min suc prob}, we consider the worst-case scenario where states following jumping trajectories yield a probability \(\text{Tr} (\rho_{T_s|t} \rho_{\text{jump}}) = 0\), i.e., the measured state \(\rho_{T_s|t}\) is orthogonal to the jump state \(\rho_{\text{jump}}\). Therefore, from \eqref{eqn: state gotten from the quantum trajectory}, the probability of obtaining the correct state satisfies the lower bound:
    \begin{align}
    \label{eqn: 28}
        \text{Tr} (\rho_{t + T_s} \rho_{T_s|t}) \geq e^{-\int_{t}^{t+T_s} \gamma(s) \, ds} \text{Tr} (\rho_d \rho_{T_s|t}),
    \end{align}
    where \(\text{Tr} (\rho_d \rho_{T_s|t}) = 1\) due to \(L^\dagger L\) being proportional to the identity.

    In this worst-case scenario, we obtain:
    \begin{align}
        \text{Tr} (\rho_{t + T_s} \rho_{T_s|t}) \geq e^{-\bar\gamma (T_s)},
    \end{align}
    since \(0 \leq\gamma(t) \leq \bar\gamma\) over the time interval \([t, t + T_s]\).
\end{proof}

In this special case, we can reformulate the stability condition as follows:
\begin{align}
    2\sqrt{1-e^{-\bar\gamma T_s}} - \lambda_0 T_s \leq 0,
\end{align}
as shown in Proposition \ref{prop: Lyapunov stable}. Additionally, the stability condition in expectation can be expressed as:
\begin{align}
    2\sqrt{1-e^{-\bar\gamma T_s}} + 1-e^{-\bar\gamma T_s} - \lambda_0 T_s \leq 0,
\end{align}
as demonstrated in Theorem \ref{thm: POVM Lyapunov stable}.

\subsection{Lower Bound in Two-Level Quantum Systems}

A two-level quantum system serves as a special case that allows for a tighter bound on the minimum probability of obtaining the nominal state. This section introduces several key cases, including the closed two-level system, as well as three common decoherence models: depolarizing decoherence, phase-damping decoherence, and amplitude-damping decoherence.

To obtain a tighter bound for the closed two-level system than the one provided by Theorem \ref{thm: general min suc prob}, we follow \cite[Theorem 2]{lee2024robust}. By generalizing this result based on the principle of mathematical induction, we derive Lemma \ref{lemma: k step}, which extends the analysis to measurements over \(l\)-step intervals and incorporates time-varying uncertainties in two-level systems.

\begin{lem}
\label{lemma: k step}
Consider the nominal system \eqref{eqn:nominal_evolution} in a two-level quantum system. If the sampling time \( T_s \) satisfies the condition \( l \bar{\Delta} T_s \leq \dfrac{\pi}{2} \), where \( \bar{\Delta} \) is the bound on the uncertainties such that \( \| H_\Delta \| = \Delta \| H'_\Delta \| \) with \( \| H'_\Delta \| = 1 \) and \( |\vec{\Delta}| \leq \bar{\Delta} \), then the probability of transitioning to the correct nominal state \( \rho_{l|k} \) satisfies the lower bound:
\begin{align}
  \operatorname{Tr} (\rho_{l|k} \rho_{k+l}) \geq \cos^2( l \bar{\Delta} T_s ).
\end{align}    
Here, \( \rho_{l|k} \) is the nominal state evolved from \( \rho_k \), while \( \rho_{k+l} \) is the state evolved under uncertainty. The transition from \( k \) to \( k+l \) assumes piecewise constant control and uncertainty over each step of length \( l \).
\end{lem}

The proof is provided in Appendix \ref{sec:proof_Lemma}.

Using Lemma \ref{lemma: k step}, Proposition \ref{prop: 2 level} can be established by dividing each step into infinitesimal intervals and applying Lemma \ref{lemma: k step} to account for time-varying Hamiltonians.

\begin{prop}
\label{prop: 2 level}
Consider the nominal system~\eqref{eqn:nominal_evolution} in a two-level quantum system.  
If the time interval between measurements \(T_s\) satisfies \(\bar{\Delta} T_s \leq \pi/2\), where \(\bar{\Delta}\) is the bound on the Hamiltonian uncertainty, \(\Delta \leq \bar{\Delta}\), then
\begin{align}
    \text{Tr}(\rho_{T_s|t} \rho_{t+T_s}) \geq \cos^2(\bar{\Delta} T_s).
\end{align}
\end{prop}

When the system is subject to environmental interactions described by~\eqref{eqn:Lindblad_equation}, we derive lower bounds on the probability of obtaining the nominal state under the combined influence of Hamiltonian uncertainty and decoherence.  
These results are summarized in Table~\ref{Tab: summarized table}, which lists the lower bounds of \(\text{Tr}(\rho_{T_s|t} \rho_{t+T_s})\) for various types of decoherence.  
For an uncertain Hamiltonian \(H_\Delta\), the bound follows from Proposition~\ref{prop: 2 level} and holds under the condition \(\bar{\Delta} T_s \leq \pi/2\).  
However, in the amplitude-damping case, a stricter condition \(\bar{\Delta} T_s \leq \pi/4\) is required.  
All cases assume that the dissipation rate is bounded by \(\gamma < \bar{\gamma}\).  
Detailed derivations are provided in Appendix~\ref{sec: MSP 2 level}.

\begin{table}[htbp]
\centering
\caption{Lower bounds on the probability of obtaining the nominal state in two-level systems.}
\resizebox{\linewidth}{!}{
\begin{tabular}{|l|c|}
\hline
Case & Lower bound of \( \text{Tr}(\rho_{T_s|t} \rho_{t+T_s}) \) \\ \hline
Closed system & \( \cos^2(l\bar\Delta T_s) \)  \\ \hline
Depolarizing decoherence &  \( \frac{1}{2} \cos^2(l\bar\Delta T_s) (1 + e^{-4\bar\gamma l T_s}) \)  \\ \hline
Phase-damping decoherence &  \( \cos^2(l\bar\Delta T_s) e^{-\bar\gamma l T_s} \) \\ \hline
Amplitude-damping decoherence &  \( \cos^2(l\bar\Delta) (1 - \bar\gamma l T_s) - \frac{1}{2} \sin(2l\bar\Delta) \) \\ \hline
\end{tabular}
}
\label{Tab: summarized table}
\end{table}

The lower bounds in Table~\ref{Tab: summarized table} can be directly applied to verify Lyapunov stability using Proposition~\ref{prop: Lyapunov stable}, and Lyapunov stability in expectation as shown in Theorem~\ref{thm: POVM Lyapunov stable}.

\section{Numerical Examples}
\label{sec:numerical_simulations}

In this section, we present quantum-system control using Q-TOPC, as described in Algorithm \ref{alg: qMPC}. The time-optimal control problem for a two-level system is first solved using PMP, following \cite{Lin2020Time}. We then compute the time-optimal control for a three-level system via gradient descent. The resulting trajectories indicate a monotonic decrease in the Lyapunov function, consistent with Corollaries \ref{cor: Lyapunov stable} and \ref{cor: Lyapunov stable nominal}, and a decrease in expectation as established in Theorem \ref{thm: POVM Lyapunov stable}. Incorporating POVM enhances fidelity compared to the unmeasured case.

\subsection{Two-Level Quantum System}
\label{sec:Simulation of a Two-Level Quantum System}
This section considers a two-level quantum system under time-optimal control derived from \eqref{eqn: OCP}. The system evolves under the free Hamiltonian \( H_0 = \sigma_z \), with a control Hamiltonian \( H_u = u(t) \sigma_x \), where \( u(t) \) is a bounded control input satisfying \(|u(t)| \leq 1\). The Pauli matrices used to represent spin operators in quantum mechanics are given by:
\begin{align*}
    \sigma_x &= \begin{pmatrix} 0 & 1 \\ 1 & 0 \end{pmatrix}, \quad
    \sigma_y = \begin{pmatrix} 0 & -i \\ i & 0 \end{pmatrix}, \quad
    \sigma_z = \begin{pmatrix} 1 & 0 \\ 0 & -1 \end{pmatrix}.
\end{align*}

Additionally, the system includes a dissipative term \( L = \sqrt{\gamma} \sigma_y \), with \( \gamma \in [0, 0.25] \). Q-TOPC is modeled for different Lindblad operators \( L \): \( L = 0 \) (closed system), \( L = 0.1 \sigma_y \), and \( L = 0.5 \sigma_y \). The cost function uses a weight \( \lambda_0 = 0.04 \). The control objective is to steer the initial state \( |0\rangle \) to the target state \( |1\rangle \), the eigenstates of \( \sigma_z \).

To verify Corollaries~\ref{cor: Lyapunov stable} and \ref{cor: Lyapunov stable nominal}, we apply Algorithm~\ref{alg: qMPC} to a system subject to repeated measurements using POVMs. A fixed time step \( T_s = 1 \) is selected, and the quantum system is measured at discrete intervals. The simulation terminates when the optimal time becomes smaller than \( T_s \), at which point the optimal time is taken as the final measurement period. Fig.~\ref{fig:cost_fidelity_once_high_error} shows a decreasing cost function along with the corresponding fidelity.

\begin{figure}[htbp]
    \centering
    \includegraphics[width=0.47\textwidth]{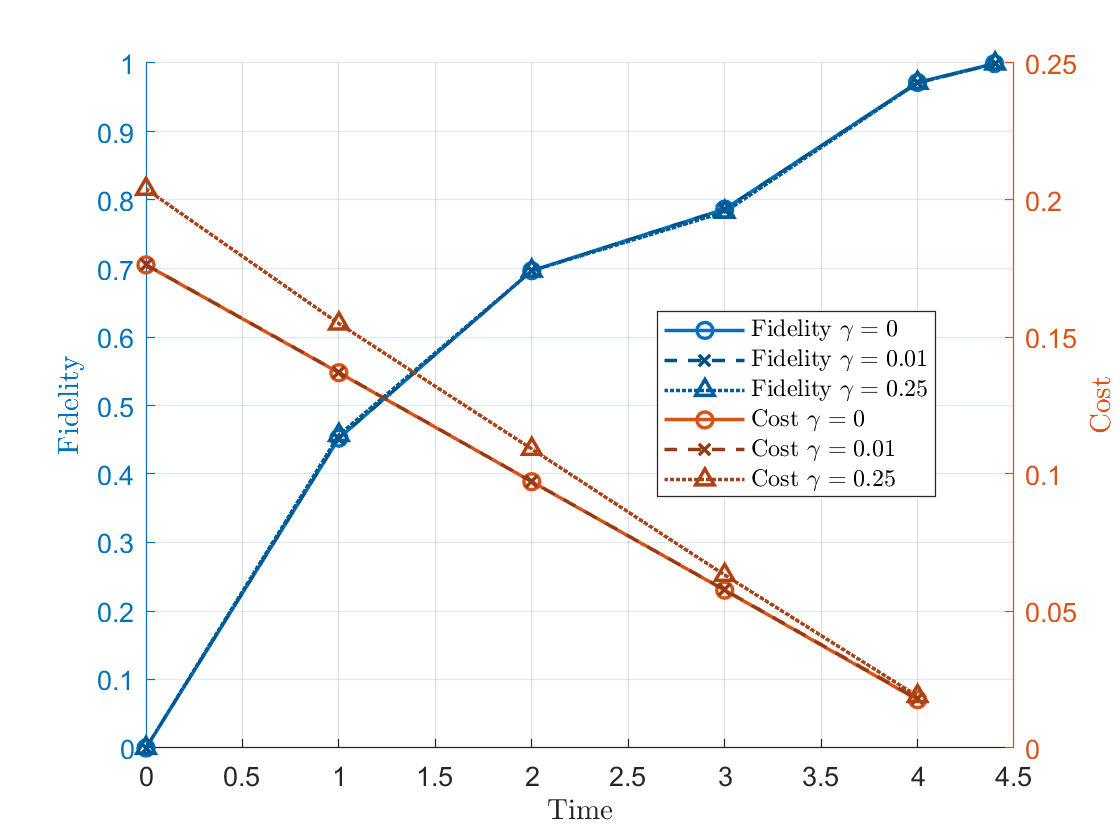}
    \caption{Evolution of the cost function and fidelity over time for a quantum trajectory under time-optimal control following Algorithm~\ref{alg: qMPC}.}
    \label{fig:cost_fidelity_once_high_error}
\end{figure}

To verify the decrease in expectation established in Theorem \ref{thm: POVM Lyapunov stable}, we conduct a Monte Carlo analysis based on $1000$ realizations, each involving $20$ measurement steps. Fig.~\ref{fig:cost_fidelity_avg_high_error} presents the averaged cost and fidelity.

\begin{figure}[htbp]
    \centering
    \includegraphics[width=0.47\textwidth]{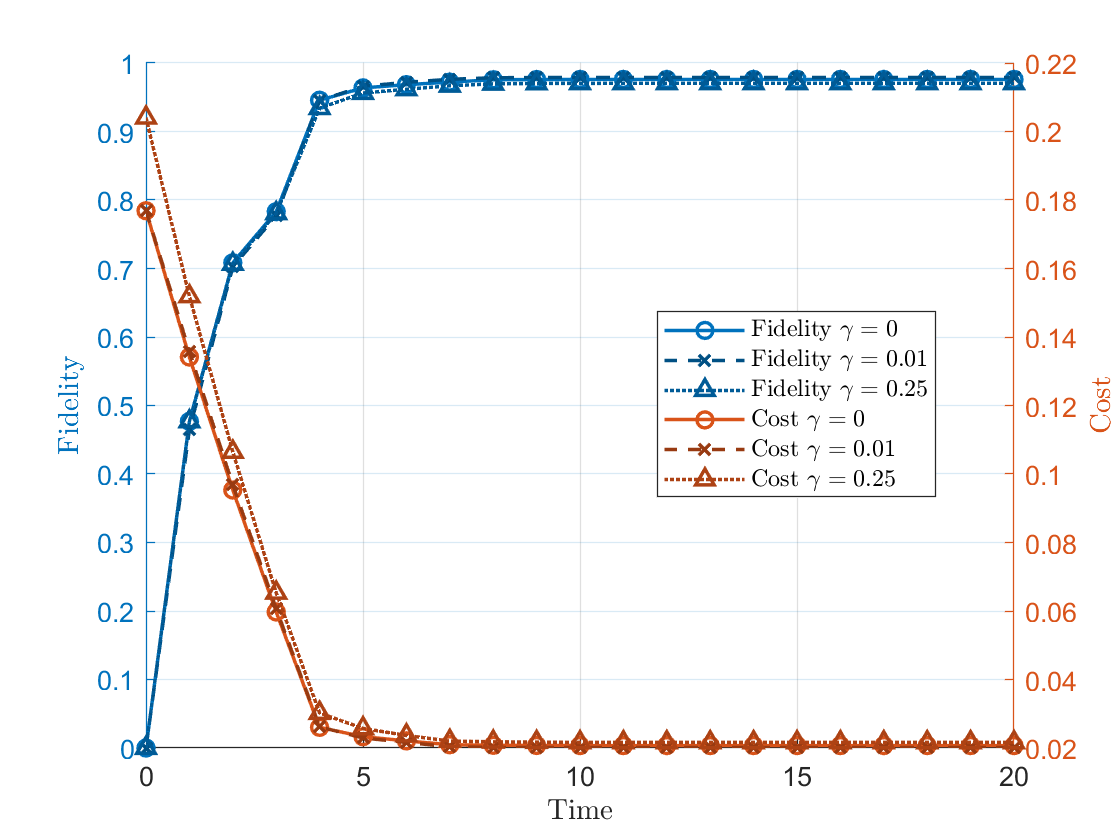}
    \caption{Averaged evolution of the cost function and fidelity for quantum trajectories under POVM control across $1000$ simulations.}
    \label{fig:cost_fidelity_avg_high_error}
\end{figure}

Lastly, we report the infidelity values in Table~\ref{tab:fidelity_comparison_high_error} to evaluate the effectiveness of incorporating measurements. As a baseline, we consider an open-loop strategy, where PMP control is applied directly to the true system without measurement feedback, resulting in an infidelity of \( 965.1 \times 10^{-4} \).

In comparison, Table~\ref{tab:fidelity_comparison_high_error} presents the infidelity values obtained when measurement feedback is incorporated. ``Nominal Meas.'' refers to the case where the measurement outcome is always assumed to match the nominal predicted state, as illustrated in Fig.~\ref{fig:cost_fidelity_once_high_error}. The labels ``Closed System,'' ``\( \gamma = 0.01 \),'' and ``\( \gamma = 0.25 \)'' indicate the nominal models used to compute the optimal control input. ``Avg. Meas.'' represents the average infidelity over $1000$ Monte Carlo simulations, where each run randomly samples a different nominal model, as shown in Fig.~\ref{fig:cost_fidelity_avg_high_error}.

Our simulations show that the simplest nominal model, as stated in Corollary~\ref{cor: Lyapunov stable nominal}, yields sufficiently effective performance. Setting \( \gamma = 0.25 \), the true upper bound, does not yield the best result, as it may lead to a pessimistic estimate and degraded time-optimality. Determining the optimal choice of \( \gamma \) remains an open question.

\begin{table}[htbp]
    \centering
     \caption{Comparison of infidelity values (\( \times 10^{-4} \)) for different measurement strategies in closed and open quantum systems.}
    \resizebox{\linewidth}{!}{
    \begin{tabular}{c|c|c}
        Infidelity (\( \times 10^{-4} \)) & Nominal Meas. & Avg. Meas. \\
        \hline
        Closed System & 8.362 & 240.0 \\
        Open System (\( \gamma = 0.01 \)) & 8.229 & 211.5  \\
        Open System (\( \gamma = 0.25 \))   & 9.174 & 296.5  \\
    \end{tabular}
    }
    \label{tab:fidelity_comparison_high_error}
\end{table}

Additionally, q-TOPC provides a more practical approach to incorporating measurements. By selecting fixed POVMs and determining the optimal time step \( T_s \) rather than using a fixed one, this framework offers a more feasible strategy for measurement preparation. In Fig.~\ref{fig:closed_optimal_fidelity_scatter}, we consider two POVM bases:

\begin{align}
\label{eqn:fixed_POVM}
\begin{aligned}
        \mathcal{M}_1 &= \left\{ |0\rangle \langle 0|, |1\rangle \langle 1| \right\}, \\
    \mathcal{M}_2 &= \left\{ \left(\frac{\sqrt{3}}{2} |0\rangle + \frac{1}{2} |1\rangle\right) \left(\frac{\sqrt{3}}{2} \langle 0| + \frac{1}{2} \langle 1|\right), \right. \\ 
    &\quad \left. \left(\frac{1}{2} |0\rangle - \frac{\sqrt{3}}{2} |1\rangle\right) \left(\frac{1}{2} \langle 0| - \frac{\sqrt{3}}{2} \langle 1|\right) \right\},
\end{aligned}
\end{align}
and q-TOPC follows the model in the absence of dissipation, i.e., $\gamma = 0$.

\begin{figure}[htbp]
    \centering
    \includegraphics[width=0.47\textwidth]{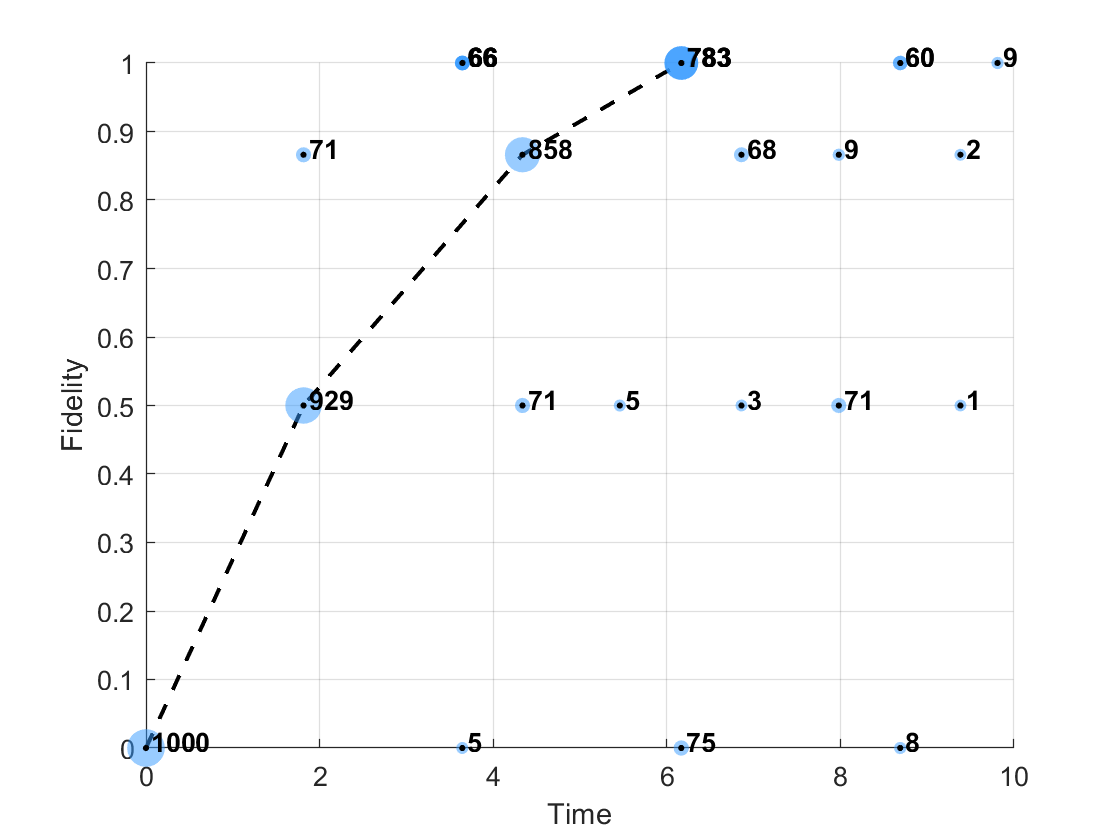}
    \caption{Bubble plot of fidelity evolution under a fixed POVM strategy~\eqref{eqn:fixed_POVM} over $1000$ simulations in a two-level quantum system. The bubble size indicates the frequency of each $(\text{time}, \text{fidelity})$ pair. A clear dominant path emerges, closely following the nominal trajectory, which occurs with much higher probability than alternative measurement outcomes. The number next to each bubble indicates the total count of occurrences for that specific pair.}
    \label{fig:closed_optimal_fidelity_scatter}
\end{figure}

\subsection{Three-Level Quantum Systems}

By extending the analysis to a three-level quantum system, we consider the free Hamiltonian \( H_0 = J_z \), with control Hamiltonian \( H_u = u(t) J_x \), where the control input satisfies \( |u(t)| \leq 1 \). The dissipative process is modeled by the Lindblad operator \( L = 0.1 J_y \). For this system, \( J_x \), \( J_y \), and \( J_z \) denote the angular momentum operators in the \( x \), \( y \), and \( z \) directions, respectively, and are defined as:

\begin{align*}
J_x &= \tfrac{1}{\sqrt{2}}\begin{pmatrix}0&1&0\\1&0&1\\0&1&0\end{pmatrix},\quad
J_y = \tfrac{i}{\sqrt{2}}\begin{pmatrix}0&-1&0\\1&0&-1\\0&1&0\end{pmatrix},\\
&\quad J_z = \operatorname{diag}(1,0,-1).
\end{align*}

We control the state from 
\[
\rho_0 = \mathrm{diag}(1,0,0)
\quad \text{to} \quad
\rho_{\text{tar}} = \mathrm{diag}(0,0,1).
\]

The simulation settings for the three-level quantum system are the same as those described in Section~\ref{sec:Simulation of a Two-Level Quantum System}, except for the system's dimensionality and the corresponding angular momentum matrices.

Figure~\ref{fig:cost_fidelity_once_dim3} illustrates the evolution of the cost function and fidelity for a quantum trajectory where the system follows the nominal trajectory under POVM control, with measurements taken at each step.

\begin{figure}[htbp]
    \centering
    \includegraphics[width=0.47\textwidth]{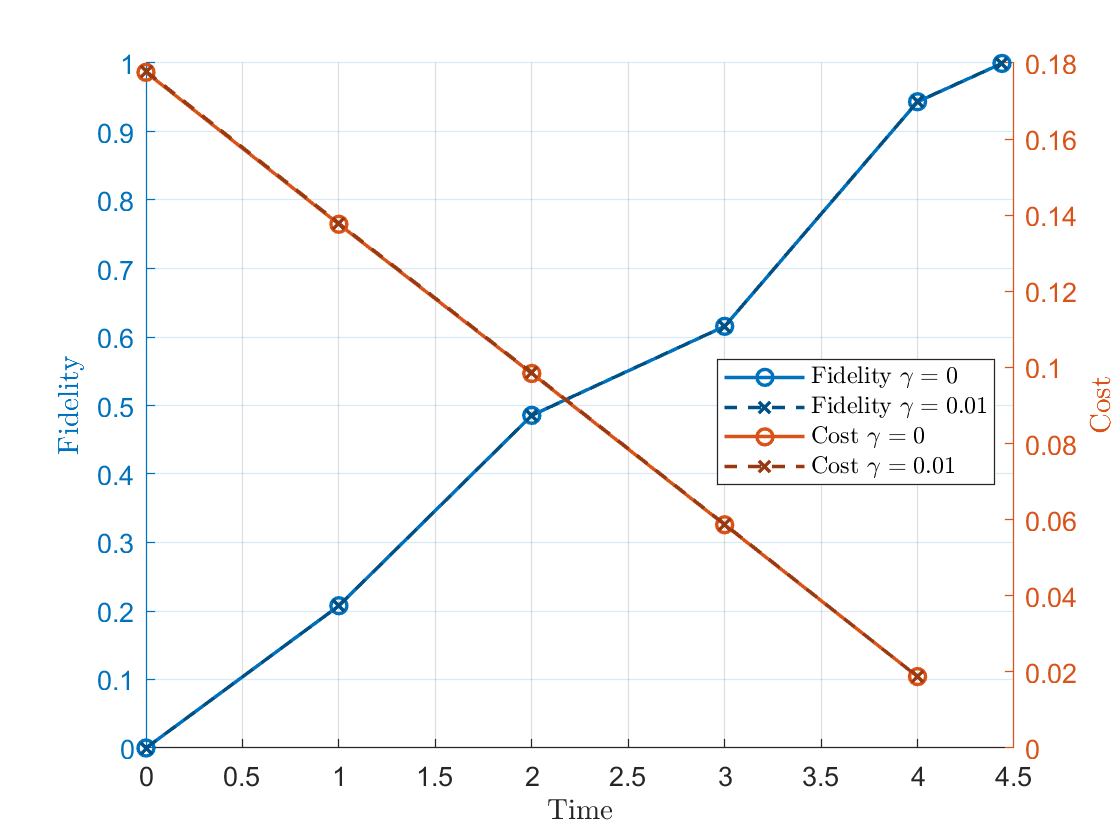}
    \caption{Cost function and fidelity evolution for a quantum trajectory under POVM control, where each step is measured along the nominal trajectory in a three-level system.}
    \label{fig:cost_fidelity_once_dim3}
\end{figure}

To verify the expected behavior, we conduct a Monte Carlo analysis based on $1000$ realizations, each involving $20$ measurement steps. Figure~\ref{fig:cost_fidelity_avg_dim3} presents the averaged evolutions of the cost function and fidelity under POVM control in the three-level system.

\begin{figure}[htbp]
    \centering
    \includegraphics[width=0.47\textwidth]{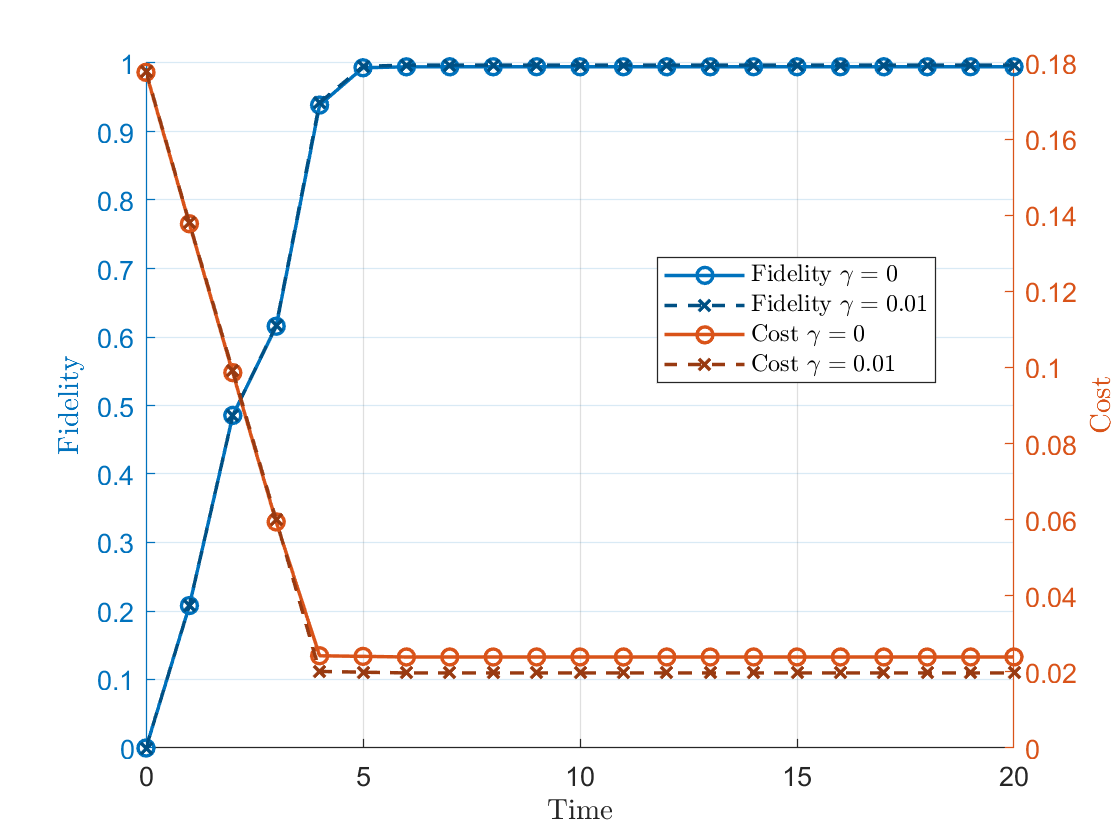}
    \caption{Averaged cost function and fidelity evolution for quantum trajectories under POVM control across $1000$ simulations in a three-level system.}
    \label{fig:cost_fidelity_avg_dim3}
\end{figure}

Table~\ref{tab:fidelity_comparison_dim3} shows the final infidelity results for the three-level system under different measurement strategies. As a baseline, directly applying PMP control to the true system without measurement feedback yields an infidelity of \( 128.4 \times 10^{-4} \).

\begin{table}[htbp]
    \centering
    \caption{Comparison of infidelity values for different measurement strategies in closed and open three-level quantum systems. }
    \begin{tabular}{c|c|c}
        Infidelity (\( \times 10^{-4} \))  & Nominal Meas. & Avg. Meas. \\
        \hline
        Closed System  & 6.315 & 54.91 \\
        Open System    & 6.122 & 30.19 \\
    \end{tabular}  
    \label{tab:fidelity_comparison_dim3}
\end{table}

\section{Conclusions and Future Work}
\label{sec:conclusion}

In this paper, we have investigated several potential conditions for controlling quantum systems using POVMs. We established Lyapunov stability in a general case and provided detailed analyses for specific scenarios, including different types of decoherence. Our results demonstrate that a fixed measurement strategy offers a feasible and systematic approach for controlling quantum states. By employing predetermined input signals and a set of predefined POVMs, the control process is simplified, facilitating the implementation of step-by-step quantum control.

\bibliographystyle{IEEEtran}
\bibliography{autosam}           



\appendix
\section{Lower Bound on the Probability for a Time-Independent Hamiltonian}
\label{sec: Appendix A}

Let the Hamiltonian \( \tilde H \) and perturbation \( \tilde H_\Delta \) be fixed during any sampling time \( T_s(k) \). The nominal evolution and the system affected by noise at the next step are given by:
\begin{align}
\label{eqn: N nominal}
    |\psi_{1|k} \rangle &= e^{-i \tilde H T_s} |\psi_k \rangle =: U |\psi_k \rangle, \\
\label{eqn: N real}
    |\psi_{k+1} \rangle &= e^{-i (\tilde H + \tilde H_\Delta) T_s} |\psi_k \rangle =: V |\psi_k \rangle,
\end{align}
where \( U \) and \( V \) are unitary operators corresponding to the nominal and perturbed evolutions, respectively.

\begin{thm}
Let 
\begin{equation}
    \Gamma = \left\| \mathbb{I} - e^{i \tilde H_\Delta T_s} \right\| + \dfrac{ T_s^2 }{2} \left\| [\tilde H,\tilde H_\Delta ] \right\|.
\end{equation}
If \(\Gamma < \sqrt{2}\), then the probability \( 1-p \) of transitioning to the nominal state \( |\psi_{1|k} \rangle \) satisfies the inequality:
\begin{equation}
    1-p \geq \left( 1 - \dfrac{1}{2} \Gamma^2 \right)^2.
\end{equation}
\end{thm}

\begin{proof}
With equations \eqref{eqn: N nominal} and \eqref{eqn: N real}, we can represent the probability \( 1-p \) of transferring to the nominal state \( |\psi_{1|k} \rangle \) as follows:

\begin{align}
\begin{aligned}
      1-p &= \left| \langle \psi_{1|k} \mid \psi_{k+1} \rangle \right|^2 = \left| \langle \psi_k \mid U^\dagger V \mid \psi_k \rangle \right|^2  \\
        &\geq \left[ \Re\left\{ \langle \psi_k \mid U^\dagger V \mid \psi_k \rangle \right\} \right]^2  \\
        &= \left( \frac{ \operatorname{Tr}\left( \rho_k (U^\dagger V + V^\dagger U)\rho_k^\dagger \right) }{2} \right)^2  \\
        &= \left\| \frac{ \rho_k (U^\dagger V + V^\dagger U)\rho_k^\dagger }{2} \right\|^2,
\end{aligned}
\end{align}
where \( \|\cdot\| \) is the spectral norm and we define \( \rho_k = | \psi_k \rangle \langle \psi_k | \).

By defining

\[
    \tilde{I} := \left\| \rho_k (U^\dagger V + V^\dagger U)\rho_k^\dagger \right\|,
\]
it can be found that
\begin{align}
\begin{aligned}
    2 - \tilde{I} &= \left\| \rho_k \left( 2 \mathbb{I} - (U^\dagger V + V^\dagger U) \right)\rho_k^\dagger \right\| \\
    &= \left\| (U - V) \rho_k \right\|^2.
\end{aligned}
\end{align}
Consequently, the probability satisfies the following inequality:
\begin{equation}
    1-p \geq \left( \frac{ 2 - \left\| (U - V) \rho_k \right\|^2 }{2} \right)^2.
\end{equation}
It is noted that \( \left\| (U - V) \rho_k \right\| = \left\| (\mathbb{I} - U^\dagger V) \rho_k \right\| \).

We now consider the triangle inequality for matrices \( A, B, C \) satisfying
\begin{align}
   \|A - B\| \leq \|C\|.
\end{align}
Applying the reverse triangle inequality yields:
\begin{align}
\begin{aligned}
     \left| \| \mathbb{I} - A \| - \| \mathbb{I} - B \| \right| &\leq \| (\mathbb{I} - A) - (\mathbb{I} - B) \| \\
     &= \| A - B \| \leq \| C \|.
\end{aligned}
\end{align}
Thus, we obtain the following bound:
\begin{equation}
\label{eqn:tri-inequality}
    \| \mathbb{I} - A \| \leq \| \mathbb{I} - B \| + \| C \|.
\end{equation}

Following the approach in~\cite{childs2019nearly}, we set \( A = U^\dagger V \), \( B = e^{i \tilde{H}_\Delta T_s} \), and \( C = \frac{T_s^2}{2} \| [\tilde{H}, \tilde{H}_\Delta] \| \), where \( U = e^{-i \tilde{H} T_s} \) and \( V = e^{-i (\tilde{H} + \tilde{H}_\Delta) T_s} \).

Substituting them into~\eqref{eqn:tri-inequality}, we find
\begin{align}
    \| (\mathbb{I} - U^\dagger V) \rho \| \leq \| \mathbb{I} - U^\dagger V \| \leq \Gamma,
\end{align}
where \( \Gamma := \left\| \mathbb{I} - e^{i \tilde{H}_\Delta T_s} \right\| + \dfrac{T_s^2}{2} \| [\tilde{H}, \tilde{H}_\Delta] \| \).

Under the condition \( \Gamma < \sqrt{2} \), the corresponding lower bound on the probability \(1 - p\) remains meaningful. Specifically, we have
\begin{align}
    1 - p \geq \left( 1 - \dfrac{1}{2} \Gamma^2 \right)^2.
\end{align}
\end{proof}


\section{Quantum Trajectory Method for Stochastic Unraveling of Quantum Dynamics}
\label{sec: Appendix B}

To introduce the stochastic unraveling of quantum dynamics \cite{yip2018quantum}, the deterministic evolution (no jump trajectory) of the system is characterized by the Schrödinger equation with a non-Hermitian Hamiltonian:
\begin{align}
\label{eqn: deterministic evolution}
    \frac{d}{d t}|\tilde{\psi_d}(t)\rangle=-i H_{\text {eff }}|\tilde{\psi_d}(t)\rangle,
\end{align}
where \(|\tilde{\psi_d}(t)\rangle\) denotes the unnormalized state vector, and the effective Hamiltonian is defined as:
\begin{align}
    H_{\text {eff }}=H(t)-\frac{i}{2} \sum_{k=1}^K L_k^{\dagger} L_k ,
\end{align}
where the value of \(K\) is determined by the number of Lindblad operators.

During an infinitesimal time increment from \(t\) to \(t+d t\), two potential evolutions are possible for \(|\tilde{\psi}(t)\rangle\): either following the deterministic path as \eqref{eqn: deterministic evolution} (with probability \(1-d p\)) or experiencing a quantum jump (with probability \(d p\)). The probability of a jump occurring is computed as:
\begin{align}
    d p=\sum_{k=1}^K\left\langle L_k^{\dagger} L_k\right\rangle d t .
\end{align}
In the event of a jump, the unnormalized state undergoes an update:
\begin{align}
    |\tilde{\psi}_{\text{jump}}(t+d t)\rangle=L_i|\tilde{\psi}(t)\rangle,
\end{align}
where \(L_i\) is selected randomly from the set \(\left\{L_k\right\}_{k=1}^K\) with the associated probability:
\begin{align}
    p_i=\left\langle\tilde{\psi}(t)\left|L_i^{\dagger} L_i\right| \tilde{\psi}(t)\right\rangle / \sum_{k=1}^K\left\langle\tilde{\psi}(t)\left|L_k^{\dagger} L_k\right| \tilde{\psi}(t)\right\rangle .
\end{align}
Thus, the system evolves to the next state as follows:
\begin{align}
\label{eqn: deterministic and jump}
    \rho (t+dt) = (1-dp) \rho_d (t+dt) +  \sum_i p_i \rho_{\text{jump},i} (t+dt),
\end{align}
where $\rho_d (t+dt) = |\psi_d (t+dt)\rangle \langle \psi_d (t+dt) |$, and $|\psi_d (t+dt)\rangle$ is the normalized state corresponding to $|\tilde\psi_d (t+dt)\rangle$. Similarly, $\rho_{\text{jump},i} (t+dt) = |\psi_{\text{jump},i} (t+dt)\rangle \langle \psi_{\text{jump},i} (t+dt) |$, where $|\psi_{\text{jump},i} (t+dt)\rangle$ is the normalized state corresponding to $|\tilde\psi_{\text{jump},i} (t+dt)\rangle$.

For discrete-time analysis over a sampling interval \(T_s\), as extensively discussed in \cite[Appendix C]{yip2018quantum}, the probability of no quantum jump occurring within a finite interval \(T_s\) (which may not be infinitesimally small) is given by:
\begin{align}
\label{eqn: no jump traj prob}
\text{Tr} (\tilde\rho(t+T_s))
     = \exp \left(-\int_t^{t+T_s} \sum_k\left\langle L_k^{\dagger} L_k\right\rangle d s \right),
\end{align}
where $\tilde\rho(t+T_s))=|\tilde{\psi}(t+T_s)\rangle\langle \tilde{\psi}(t+T_s)|$.

\section{Proof of Lemma \ref{lemma: k step}}
\label{sec:proof_Lemma}
We prove the expectation $\langle \psi_{l|k}|\psi_{k+l}\rangle$ for $l$-step, satisfying
\begin{align}
    |\langle \psi_{l|k}|\psi_{k+l}\rangle| \geq \cos (l \bar \Delta T_s) ,
\end{align}
by utilizing the principle of mathematical induction. 

For convenience, the nominal state evolution is represented by:
\begin{align}
   U_i = e^{-i H_i T_s},
\end{align}
and the real evolution is:
\begin{align}
    W_i = e^{-i (H_i+H_{\Delta,i}) T_s},
\end{align}
where \(i\) represents the next $i$-step evolution.
Therefore, the overall expectation value is:
\begin{align}
    E(l)=\langle \psi_k | U_1^{\dagger} U_2 ^{\dagger}  \cdots U_l^{\dagger}W_l \cdots W_2 W_1|\psi_k \rangle .
\end{align}

\textbf{Base Case:} 
Firstly, we discuss the case when $l=2$. The expectation can be extended as 
\begin{align}
\begin{aligned}
       \langle \psi_{2|k}|\psi_{k+2}\rangle=&\langle \psi_{k}| U_1 U_2 W_2 W_1|\psi_{k}\rangle\\
       =&\langle \psi_{2|k}| W_2 |\psi_{1|k}\rangle \langle \psi_{1|k}| W_1|\psi_{k}\rangle\\
       &+\langle \psi_{2|k}| W_2 |\psi_{1|k}^{\perp}\rangle \langle \psi_{1|k}^{\perp}| W_1|\psi_{k}\rangle ,
\end{aligned}
\end{align}
since we have known $|\langle \psi_{1|k}| W_1|\psi_{k}\rangle|>\cos({\bar{\Delta}T_s})$, where $\omega_1, \xi_1 \in \mathbb{R}$ represent the phase. Besides, in the two-level system, when states are orthogonal, they satisfy  
\begin{align}
    | \langle \psi_{1|k}| W_1|\psi_{k}\rangle |^2 + | \langle \psi_{1|k}^{\perp}| W_1|\psi_{k}\rangle |^2 = 1.
\end{align}
We can assume $\langle \psi_{1|k}| W_1|\psi_{k}\rangle = \cos\phi_1 e^{i\omega_1}$, then $\langle \psi_{1|k}^{\perp}| W_1|\psi_{k}\rangle =\sin\phi_1 e^{i\xi_1}$.
The same idea can be assumed to $\langle \psi_{2|k}| W_2 |\psi_{1|k}\rangle = \cos\phi_2 e^{i\omega_2}$, $\langle \psi_{2|k}| W_2 |\psi_{1|k}^{\perp}\rangle = \pm \sin\phi_2 e^{i\xi_2}$.
Therefore, 
\begin{align}
|\langle \psi_{2|k}|\psi_{k+2}\rangle |
   \!=\! \big| \! \cos\!\phi_2 \! \cos\!\phi_1 e^{i(\!\omega_2\!+\!\omega_1\!)}
     \!+\! \sin\!\phi_2 \! \sin\!\phi_1 e^{i ( \!\xi_2\!+\!\xi_1 \!)} \! \big| .
\end{align}
If we take the first term and the second term as opposite directions, they are the smallest value by triangle inequality. Besides, since $|\langle \psi_{l|k}| W_l|\psi_{l-1|k}\rangle|>\cos({\bar{\Delta}T_s})$, we can assume $\phi_1, \phi_2 \leq \bar {\Delta} T_s$. Therefore, 
\begin{align}
\begin{aligned}
       &|\langle \psi_{2|k}|\psi_{k+2}\rangle| \geq \cos (\phi_2 + \phi_1) \geq \cos (2 \bar {\Delta}T_s) .
\end{aligned}
\end{align}

\textbf{Inductive Step:} We assume that $|\langle \psi_{l|k}|\psi_{k+l}\rangle|$ is true for some arbitrary positive integer \(l\), i.e.,
\begin{align}
    |\langle \psi_{l|k}|\psi_{k+l}\rangle| \geq \cos (l \bar \Delta T_s).
\end{align}
As in the base case, we can assume 
\begin{align}
    \langle \psi_{l|k}|\psi_{k+l}\rangle=\cos (\sum_{i=1}^l \phi_i) e^{i \sum_{i=1}^{l+1} \omega_i},
\end{align} 
and the corresponding orthogonal measurement 
\begin{align}
    \langle \psi_{l|k}^{\perp}|\psi_{k+l}\rangle = \sin (\sum_{i=1}^l \phi_i) e^{i \sum_{i=1}^l \xi_i}
\end{align}
by selecting all $\phi_i < \bar \Delta T_s$.

Then, considering the scenario for \(l+1\), we compute
\begin{align}
\begin{aligned}
     \langle \psi_{l+1|k}|\psi_{k+l+1}\rangle=&\langle \psi_{l+1|k} \prod_{j=0}^{l} W_{l+1-j}|\psi_{k}\rangle\\
   =&\langle \psi_{l+1|k}| W_{l+1} |\psi_{l|k}\rangle \langle \psi_{l|k}|\psi_{k+l}\rangle\\
   &+\langle \psi_{l+1|k}| W_{l+1} |\psi_{l|k}^{\perp}\rangle \langle \psi_{l|k}^{\perp}|\psi_{k+l}\rangle\\
   =& \left|\cos\phi_{l+1} \cos (\sum_{i=1}^l \phi_i) e^{i \sum_{i=1}^{l+1} \omega_i} \right. \\
   &\quad + \left.\sin\phi_{l+1} \sin (\sum_{i=1}^l \phi_i) e^{i \sum_{i=1}^{l+1} \xi_i} \right|\\
   \geq& \cos (\sum_{i=1}^{l+1} \phi_i).
\end{aligned}
\end{align}

By the assumption \( \phi_i \leq \bar{\Delta} T_s \) for all \(i\), we can observe
\begin{align}
    |\langle \psi_{l+1|k}|\psi_{k+l+1}\rangle| \geq \cos ((l+1) \bar \Delta T_s).
\end{align}
Therefore, we establish that if \(|\langle \psi_{l|k}|\psi_{k+l}\rangle|\) holds, then \(|\langle \psi_{l+1|k}|\psi_{k+l+1}\rangle|\) holds as well.

Hence, by the principle of mathematical induction, it follows that \(|\langle \psi_{l+1|k}|\psi_{k+l+1}\rangle|\) holds for all positive integers \(l\), under the condition that \(l \bar{\Delta} T_s \leq \frac{\pi}{2}\).

\section{Minimum Success Probability for Different Decoherence}
\label{sec: MSP 2 level}

\subsection{Depolarizing decoherence}

In a two-level quantum system with depolarizing decoherence, the probability of transferring to the correct nominal state $\rho_{T_s|t}$ satisfies the lower bound:
\begin{align}
    \label{eqn: distance prob. DD}
    \text{Tr} (\rho_{t+T_s} \rho_{T_s|t}) \geq \frac{1}{2} + \frac{1}{2} e^{- \bar{\gamma} T_s},
\end{align}
where \(\bar{\gamma}\) is an upper bound on the Lindblad coefficients such that \(\gamma(t) \leq \bar{\gamma}\) for all times \(t\). The state \(\rho_{t+T_s}\) evolves from \(\rho_t\) according to the Lindblad equation \eqref{eqn:Lindblad_equation}.

If the system includes an uncertain Hamiltonian \( H_\Delta \), the probability of transitioning to the correct nominal state satisfies:
\begin{align}
\label{eqn: distance prob. uncertain DD}
\operatorname{Tr}\left( \rho_{t+T_s} \, \rho_{T_s|t} \right) \geq \frac{1}{2} \cos^2( \bar{\Delta} T_s ) \left( 1 + e^{- \bar{\gamma} T_s } \right)
\end{align}
under the conditions \( \gamma(t) \leq \bar{\gamma} \) and \( \bar{\Delta} T_s \leq \dfrac{\pi}{2} \). This probability has been proven in \cite{Lee2024Model}.

\subsection{Phase-Damping Decoherence}
Phase-damping decoherence in a two-level quantum system can be described by a Lindblad operator \(L\) satisfying
\begin{align}
    L^\dagger L = \gamma(t)\mathbb{I}.
\end{align}
In this case, the probability of transferring to the nominal state $\rho_{T_s|t}$ obeys
\begin{align}
\label{eqn: distance prob. PD}
    \mathrm{Tr}(\rho_{t+T_s}\rho_{T_s|t}) \;\geq\; e^{-\bar\gamma T_s},
\end{align}
whenever $\gamma(t)\leq \bar\gamma$.

If the two-level quantum system with phase-damping decoherence also includes an uncertain Hamiltonian \(H_\Delta\), the probability of transferring to the correct nominal state $\rho_{T_s|t}$ satisfies the following lower bound:
\begin{align}
\label{eqn: distance prob. uncertain PD}
    \text{Tr} (\rho_{t+T_s} \rho_{T_s|t}) \geq \cos^2(\bar{\Delta} T_s ) e^{-\bar\gamma T_s}
\end{align}
under the conditions that $\gamma(t)\leq \bar \gamma$ and $\bar{\Delta} T_s  < \pi/2$.

\begin{proof}
    We assume that a state $\rho'_{t + T_s}$ is obtained from the system evolving only under the uncertain Hamiltonian, where
    \begin{align}
        \rho'_{t + T_s} = U_\Delta \rho_t U^\dagger_\Delta,
    \end{align}
    and 
    \begin{align}
        U_\Delta = e^{-i(H + H_\Delta) T_s}.
    \end{align}
    Then,
    as described in \eqref{eqn: 28}, we can establish a lower bound for the probability of transferring to the correct nominal state as follows:
    \begin{align}
    \label{eqn: 32}
    \begin{aligned}
        \text{Tr} (\rho_{t + T_s} \rho_{T_s|t}) &\geq e^{- \int_{t}^{t + T_s} \gamma(s) \, ds} \text{Tr} (\rho'_{t+T_s} \rho_{T_s|t}) \\
        &\geq \cos^2(\bar{\Delta} T_s) e^{-\bar\gamma T_s },
    \end{aligned}
    \end{align}
    where $\rho'_{t + T_s}$ denotes the state corresponding to the deterministic trajectory under the uncertain Hamiltonian. Therefore, applying the result from Proposition \ref{prop: 2 level}, the final inequality in \eqref{eqn: 32} is obtained.
\end{proof}

\subsection{Amplitude-Damping Decoherence}
In the two-level quantum system with amplitude-damping decoherence, the probability of transferring to the correct nominal state $\rho_{T_s|t}$ satisfies
\begin{align}
\label{eqn: distance prob. AD}
    \text{Tr} (\rho_{t + T_s} \rho_{T_s|t}) \geq 1 - \bar{\gamma} T_s
\end{align}
under the condition $\gamma(t) \leq \bar{\gamma}$.

\begin{cor}
    If the two-level quantum system with amplitude-damping decoherence also includes an uncertain Hamiltonian \( H_\Delta \), the probability of transferring to the correct nominal state $\rho_{T_s|t}$ satisfies the lower bound
\begin{align}
\label{eqn: distance prob. uncertain AD}
    \text{Tr} (\rho_{t + T_s} \rho_{T_s|t}) \geq \cos^2(\bar{\Delta} T_s ) (1 - \bar{\gamma} T_s) - \frac{1}{2} \sin(2 \bar{\Delta}T_s)
\end{align}
under the condition $\gamma(t) \leq \bar{\gamma}$ and \(\bar{\Delta}T_s  \leq \pi/4\).
\end{cor}

\begin{proof}
    We assume a state $\rho'_{t + T_s} = |\psi'_{t + T_s}\rangle \langle \psi'_{t + T_s}|$, which is obtained from
    \begin{align}
        |\psi'_{t + T_s}\rangle = e^{-i \int_{t}^{t + T_s} (H(s) + H_\Delta(s)) \, ds} |\psi_t\rangle.
    \end{align}
    Then, we add this pseudo-measurement to obtain
    \begin{align}
    \label{eqn: expand probability}
    \begin{aligned}
        & \text{Tr} (\rho_{t + T_s} \rho_{T_s|t}) = \langle \psi_{T_s|t} | \rho_{t + T_s} | \psi_{T_s|t} \rangle \\
        & = \langle \psi_{T_s|t} |\psi'_{t + T_s}\rangle \langle \psi'_{t + T_s} | \rho_{t + T_s} | \psi'_{t + T_s} \rangle \langle \psi'_{t + T_s} | \psi_{T_s|t} \rangle \\
        & + \langle \psi_{T_s|t} |\psi'_{\perp,t + T_s} \rangle \langle \psi'_{\perp,t + T_s} | \rho_{t + T_s} | \psi'_{\perp,t + T_s} \rangle \langle \psi'_{\perp,t + T_s} | \psi_{T_s|t} \rangle \\
        & + \langle \psi_{T_s|t} |\psi'_{\perp,t + T_s} \rangle \langle \psi'_{\perp,t + T_s} | \rho_{t + T_s} | \psi'_{t + T_s} \rangle \langle \psi'_{t + T_s} | \psi_{T_s|t} \rangle \\
        & + \langle \psi_{T_s|t} |\psi'_{t + T_s} \rangle \langle \psi'_{t + T_s} | \rho_{t + T_s} | \psi'_{\perp,t + T_s} \rangle \langle \psi'_{\perp,t + T_s} | \psi_{T_s|t} \rangle.
    \end{aligned}
    \end{align}
    Here, we assume 
    \begin{align}
        \langle \psi_{T_s|t} |\psi'_{t + T_s}\rangle = \cos(\phi) e^{i\omega}
    \end{align} and
    \begin{align}
        \langle \psi_{T_s|t} |\psi'_{\perp,t + T_s}\rangle = \sin(\phi) e^{i\xi},
    \end{align}
    where $\phi \leq T_s \bar{\Delta}$. Then, we assume that the orthogonal state is 
    \begin{align}
        |\psi'_{\perp,t + T_s}\rangle \langle \psi'_{\perp,t + T_s}| = \mathbb{I} - |\psi'_{t + T_s}\rangle \langle \psi'_{t + T_s}|.
    \end{align}
    
    Since \begin{align}
        \langle \psi'_{t + T_s} | \rho_{t + T_s} | \psi'_{t + T_s} \rangle \geq 1 - \bar{\gamma}  T_s,
    \end{align} 
    we may assume there is a $p'_\gamma \leq \bar{\gamma} T_s$ such that 
    \begin{align}
        \langle \psi'_{t + T_s} | \rho_{t + T_s} | \psi'_{t + T_s} \rangle = 1 - p'_\gamma
    \end{align}
    and
    \begin{align}
        \langle \psi'_{\perp,t + T_s} | \rho_{t + T_s} | \psi'_{\perp,t + T_s} \rangle = p'_\gamma.
    \end{align}
    Then, the terms in \eqref{eqn: expand probability} can be bounded as:
    \begin{align}
    \begin{aligned}
        & \langle \psi_{T_s|t} | \psi'_{\perp,t + T_s} \rangle \langle \psi'_{\perp,t + T_s} | \rho_{t + T_s} | \psi'_{t + T_s} \rangle \langle \psi'_{t + T_s} | \psi_{T_s|t} \rangle \\
        & + \langle \psi_{T_s|t} | \psi'_{t + T_s} \rangle \langle \psi'_{t + T_s} | \rho_{t + T_s} | \psi'_{\perp,t + T_s} \rangle \langle \psi'_{\perp,t + T_s} | \psi_{T_s|t} \rangle \\
        & = 2\Re \left( \langle \psi_{T_s|t} | \psi'_{\perp,t + T_s} \rangle \langle \psi'_{\perp,t + T_s} | \rho_{t + T_s} | \psi'_{t + T_s} \rangle \langle \psi'_{t + T_s} | \psi_{T_s|t} \rangle \right) \\
        & \leq 2 \|\langle \psi_{T_s|t} | \psi'_{\perp,t + T_s} \rangle \langle \psi'_{\perp,t + T_s} | \rho_{t + T_s} | \psi'_{t + T_s} \rangle \langle \psi'_{t + T_s} | \psi_{T_s|t} \rangle\| \\
        & = 2 \sqrt{(1 - p'_\gamma) p'_\gamma} \cos(\phi) \sin(\phi) \leq \frac{1}{2} \sin(2 \phi).
    \end{aligned}
    \end{align}
    Then, the probability of transferring to the correct quantum state can be bounded as:
    \begin{align}
    \begin{aligned}
        & \text{Tr} (\rho_{t + T_s} \rho_{T_s|t}) \\
        & \geq \left( \cos^2(\phi) (1 - p'_\gamma) + \sin^2(\phi) p'_\gamma \right) - \frac{1}{2} \sin(2 \phi) \\
        & \geq \cos^2(\bar{\Delta} T_s) (1 - \bar{\gamma} T_s) - \frac{1}{2} \sin(2 \bar{\Delta} T_s ).
    \end{aligned}
    \end{align}
    To obtain the last inequality, we use the assumption \(\bar{\Delta}T_s   \leq \pi/4\). Therefore, the lower bound on the probability of transferring to the correct nominal is established.
\end{proof}


\end{document}